\documentclass[10pt,conference]{IEEEtran}
\IEEEoverridecommandlockouts

\usepackage{url}

\usepackage{amsthm, amssymb}
\usepackage[cmex10]{amsmath}

\usepackage{amsfonts}
 \usepackage{setspace}
\usepackage{latexsym}
\usepackage{verbatim}
\usepackage{tikz}
\usetikzlibrary{positioning,arrows}

\usepackage{xspace}
\usepackage{graphicx}
\usepackage{paralist}
\usepackage{mathtools}
\usepackage{calc}
\usepackage{cite}

\newtheorem{proposition}{Proposition}[section]
\newtheorem{definition}{Definition}[section]
\newtheorem{remark}{Remark}[section]

\newtheorem{lemma}{Lemma}[section]
\newtheorem{theorem}{Theorem}[section]

\newcommand{\details}[1]{}

\newcommand{\AAWA}{\text{AAWA}}
\newcommand{\NAAWA}{\text{$n$AAWA}}
\newcommand{\HLTL}{\text{HyperLTL}}
\newcommand{\AHLTL}{\text{AHyperLTL}}
\newcommand{\SLTL}{\text{LTL$_{S}$}}
\newcommand{\SHLTL}{\text{HyperLTL$_{S}$}}
\newcommand{\CHLTL}{\text{HyperLTL$_{C}$}}
\newcommand{\HPDL}{\text{HyperPDL$-\Delta$}}
\newcommand{\HQPTL}{\text{HyperQPTL}}
\newcommand{\HU}{\text{H$_\mu$}}
\newcommand{\FOE}{\text{FOL[$<$,E]}}
\newcommand{\MOE}{\text{S1S[E]}}

\newcommand{\NBA}{\text{NBA}}

\newcommand{\PDL}{\text{PDL}}
\newcommand{\QPTL}{\text{QPTL}}
\newcommand{\CTL}{\text{CTL}}
\newcommand{\CTLStar}{\text{CTL$^{*}$}}
\newcommand{\LTL}{\text{LTL}}
\newcommand{\HCTLStar}{\text{HyperCTL$^{*}$}}

\newcommand{\Lang}{\mathcal{L}}
\newcommand{\AP}{\textsf{AP}}
\newcommand{\Ku}{\mathcal{K}}
\newcommand{\Instance}{\mathcal{I}}
\newcommand{\Au}{\mathcal{A}}

\newcommand \tpl[1]{\langle #1 \rangle}
 \newcommand{\Var}{\textsf{VAR}}
  \newcommand{\Lab}{\textit{Lab}}
 \newcommand{\stfr}{\textit{stfr}}
\newcommand{\SUCC}{\textit{succ}}

\newcommand{\acc}{\textit{acc}}
\newcommand{\Acc}{\textit{Acc}}
\newcommand{\cl}{{\textit{cl}}}
\newcommand{\Dom}{{\textit{Dom}}}
\newcommand{\dir}{{\textit{dir}}}

\newcommand{\last}{{\textit{lst}}}
\newcommand{\first}{{\textit{fst}}}

\newcommand{\Inst}{{\textit{L}}}
\newcommand{\halt}{{\textit{halt}}}
\newcommand{\init}{{\textit{init}}}
\newcommand{\Beg}{\textit{beg}}

\newcommand{\Tower}{\mathsf{Tower}}

\newcommand{\inc}{{\mathsf{inc}}}
\newcommand{\dec}{{\mathsf{dec}}}
\newcommand{\zero}{{\mathsf{zero}}}
\newcommand{\instr}{{\textit{op}}}
\newcommand{\From}{{\textit{from}}}
\newcommand{\To}{{\textit{to}}}

\newcommand{\Until}{\mathbin{\textsf{U}}}

\newcommand{\Release}{\textsf{R}}

\newcommand{\Next}{\textsf{X}}

\newcommand{\Always}{\textsf{G}}

\newcommand{\Eventually}{\textsf{F}}

\def\N{{\mathbb{N}}}
\def\B{{\mathbb{B}}}
\def\S{{\mathcal{S}}}
\def\C{{\mathcal{C}}}
\def\F{{\mathcal{F}}}

\newcommand{\true}{\texttt{true}}
\newcommand{\false}{\texttt{false}}

\def\PSPACE{{\sc Pspace}}
\def\EXPSPACE{{\sc Expspace}}
\def\NLOGSPACE{{\sc Nlogspace}}

\newcommand{\Rel}[2]{\ensuremath{#1[#2]}}

\newcommand{\Into}{\mathrel{\rightarrow}} 
\newcommand{\ie}{i.e.\xspace}
\newcommand{\vect}[1]{\ensuremath{\overline{#1}}}
\newcommand{\DefinedAs}{\ensuremath{\,\stackrel{\text{\textup{def}}}{=}\,}}

\begin{document}


\newcommand{\Thanks}{\thanks{This work was funded in part by the
    Madrid Regional Government under project ``S2018/TCS-4339
    (BLOQUES-CM)'', by Spanish National Project ``BOSCO
    (PGC2018-102210-B-100)''.}}

\title{Asynchronous Extensions of HyperLTL\Thanks}


\IEEEpubid{\makebox[\columnwidth]{978-1-6654-4895-6/21/\$31.00~
\copyright2021 IEEE \hfill} \hspace{\columnsep}\makebox[\columnwidth]{
}}

\author{\IEEEauthorblockN{Laura Bozzelli}
  \IEEEauthorblockA{University of Napoli ``Federico II''\\Napoli, Italy}
\and
\IEEEauthorblockN{Adriano Peron}
\IEEEauthorblockA{University of Napoli ``Federico II''\\Napoli, Italy}
\and
\IEEEauthorblockN{C\'esar S\'anchez}
\IEEEauthorblockA{IMDEA Software Institute\\Madrid, Spain}}

\maketitle

\begin{abstract}
  Hyperproperties are a modern specification paradigm that extends
  trace properties to express properties of sets of traces.
  Temporal logics for hyperproperties studied in the literature,
  including $\HLTL$, assume a synchronous semantics and enjoy a
  decidable model checking problem.
  In this paper, we introduce two asynchronous and orthogonal
  extensions of $\HLTL$, namely \emph{Stuttering $\HLTL$} ($\SHLTL$)
  and \emph{Context $\HLTL$} ($\CHLTL$).
  Both of these extensions are useful, for instance, to formulate
  asynchronous variants of information-flow security properties.
  We show that for these logics, model checking is in general
  undecidable.
  On the positive side, for each of them, we identify a fragment with
  a decidable model checking that subsumes $\HLTL$ and that can
  express meaningful asynchronous requirements.
  Moreover, we provide the exact computational complexity of model
  checking for these two fragments which, for the $\SHLTL$ fragment,
  coincides with that of the strictly less expressive logic $\HLTL$.
  %
\end{abstract}

\IEEEpeerreviewmaketitle

\section{Introduction}\label{sec:Intro}

Model checking is a well-established formal method technique to
automatically check for global correctness of finite-state
systems~\cite{Clarke81ctl,Queille81verification}.
Properties for model checking are usually specified in classic
\emph{regular} temporal logics such as \LTL, \CTL, and
\CTLStar~\cite{Pnueli77,EmersonH86}, which provide temporal modalities
for describing the ordering of events along individual execution
traces of a system (trace properties).
These logics lack mechanisms to relate distinct traces, which is
required to express important information-flow security policies.
Examples include properties that compare observations made by an
external low-security agent along traces resulting from different
values of not directly observable inputs.
These security requirements go, in general, beyond regular properties.

In the last decade, a novel specification paradigm has been introduced
that generalizes traditional regular trace properties by properties of
sets of traces, the so called
\emph{hyperproperties}~\cite{ClarksonS10}.
Hyperproperties relate distinct traces and are useful to formalize
information-flow security policies like
noninterference~\cite{goguen1982security,McLean96} and observational
determinism~\cite{ZdancewicM03}.
Hyperproperties also have applications in other settings, such as the
symmetric access to critical resources in distributed
protocols~\cite{FinkbeinerRS15}.
Many temporal logics for hyperproperties have been proposed in the
literature~\cite{DimitrovaFKRS12,ClarksonFKMRS14,BozzelliMP15,Rabe2016,FinkbeinerH16,CoenenFHH19,GutsfeldMO20}
for which model checking is decidable, including
$\HLTL$~\cite{ClarksonFKMRS14}, $\HCTLStar$~\cite{ClarksonFKMRS14},
$\HQPTL$~\cite{Rabe2016,CoenenFHH19}, and $\HPDL$~\cite{GutsfeldMO20}
which extend $\LTL$, $\CTLStar$, $\QPTL$~\cite{SistlaVW87}, and
$\PDL$~\cite{FischerL79}, respectively, by explicit first-order
quantification over traces and trace variables to refer to multiple
traces at the same time.
%

In all these logics, the mechanism for comparing distinct traces is
synchronous and consists in evaluating the temporal modalities by a
lockstepwise traversal of all the traces assigned to the quantified
trace variables.
This represents a limitation in various
scenarios~\cite{Finkbeiner17,GutsfeldOO21} where properties of
interest are instead asynchronous, since these properties require to
relate traces at distinct time points which can be arbitrarily far
from each other.
Recently, two powerful and expressively equivalent formalisms have
been introduced~\cite{GutsfeldOO21} for specifying asynchronous
linear-time hyperproperties.
The first one, called $\HU$, is based on a fixpoint calculus, while
the second one exploits parity multi-tape Alternating Asynchronous
Word Automata (\AAWA)~\cite{GutsfeldOO21} for expressing the
quantifier-free part of a
specification.
%
%
$\AAWA$ allow to specify very expressive non-regular multi-trace
properties.
As a matter of fact, model checking against $\HU$ or its $\AAWA$-based
counterpart is undecidable even for the (quantifier) alternation-free
fragment.
In~\cite{GutsfeldOO21}, two decidable subclasses of parity $\AAWA$ are
identified which lead to $\HU$ fragments with decidable model
checking.
%
%
Both subclasses express only $\omega$-regular languages over the
synchronous product of tuples of traces with fixed arity.
In particular, the first subclass captures all the multi-trace regular
properties and the corresponding $\HU$ fragment (the so called
$k$-synchronous fragment for a given $k\geq 1$) is strictly more
expressive than $\HLTL$, while the second subclass is non-elementarily
more succinct than the first subclass and leads to a $\HU$ fragment
which seems expressively incomparable with $\HLTL$.

\paragraph*{Our contribution}
In this paper, we introduce two novel expressive extensions of $\HLTL$
for the specification of asynchronous linear-time hyperproperties,
obtained by adding intuitive logical features that provide natural
modeling facilities.
The first formalism, that we call \emph{Stuttering $\HLTL$}
($\SHLTL$), is useful in information-flow security settings where an
observer is not time-sensitive, \ie the observer cannot distinguish
consecutive time points along an execution having the same
observation.
This requires asynchronously matching sequences of observations along
distinct execution traces.
The novel feature of $\SHLTL$ consists in temporal modalities
parameterized by finite sets $\Gamma$ of $\LTL$ formulas.
These modalities are evaluated along sub-traces of the given traces
which are obtained by removing ``redundant'' positions with respect to
the pointwise evaluation of the $\LTL$ formulas in $\Gamma$.
We show that model checking against the alternation-free fragment of
$\SHLTL$ is already undecidable.
On the positive side, we identify a meaningful fragment, called
\emph{simple $\SHLTL$}, with a decidable model-checking problem, which
strictly subsumes $\HLTL$ and allows to express asynchronous variants
of relevant security properties such as
noninterference~\cite{goguen1982security} and observational
determinism~\cite{ZdancewicM03}.
Moreover, model checking against simple $\SHLTL$ has the same
computational complexity as model checking for $\HLTL$ and is
expressively incomparable with the two $\HU$ fragments previously
described
In particular, unlike these two fragments and $\HLTL$, quantifier-free
formulas of simple $\SHLTL$ can express some non-regular multi-trace
properties.

The second logic that we introduce, called \emph{Context $\HLTL$}
($\CHLTL$), allows to specify complex combinations of asynchronous and
synchronous requirements.
$\CHLTL$ extends $\HLTL$ by unary modalities parameterized by a
non-empty subset $C$ of trace variables (\emph{context}) which
restrict the evaluation of the temporal modalities to the traces
associated with the variables in $C$.
Like $\SHLTL$,  model checking against $\CHLTL$ is undecidable.
In this case we exhibit a fragment of $\CHLTL$ which is, in a certain
sense, maximal with respect to the decidability of model checking, and
extends $\HLTL$ by allowing the comparison of different traces at time
points of bounded distance.
This fragment is subsumed by $k$-synchronous $\HU$, and we establish
that for a fixed quantifier alternation depth, model checking this
fragment is exponentially harder than model checking $\HLTL$.

With regard to expressiveness issues, both $\CHLTL$ and $\SHLTL$ are
subsumed by $\HU$.
On the other hand, questions concerning the comparison of the
expressive power of $\SHLTL$ and $\CHLTL$ are left open: we conjecture
that (simple) $\SHLTL$ and $\CHLTL$ are expressively incomparable.

\paragraph*{Related work}
Another linear-time temporal logic, called asynchronous $\HLTL$
($\AHLTL$), for pure asynchronous hyperproperties and useful for
asynchronous security analysis has been recently introduced
in~\cite{CoenenFS21}.
This logic, which is expressively incomparable with $\HLTL$, adds an
additional quantification layer over the so called trajectory
variables.
Intuitively, a \emph{trajectory} describes an asynchronous
interleaving of the traces in the current multi-trace where single
steps of distinct traces can overlap, and temporal modalities, indexed
by trajectory variables, are evaluated along the associated
trajectories.
The logic has an undecidable model-checking problem,
but~\cite{CoenenFS21} identifies practical fragments with decidable
model-checking, and reports an empirical evaluation. 
 
Other known logics for linear-time hyperproperties are the first-order
logic with equal-level predicate $\FOE$~\cite{Finkbeiner017} and its
monadic second-order extension $\MOE$~\cite{CoenenFHH19}.
We conjecture that these logics are expressively incomparable with
$\HU$, $\CHLTL$, and $\SHLTL$.
For instance, we believe that $\MOE$ cannot express counting
properties requiring that two segments along two different traces at
an unbounded distance from each other have the same length.
This kind of requirements can be instead expressed in $\CHLTL$ and
$\HU$.
Proving these conjectures are left for future work.


\section{Preliminaries}
\label{sec-prelim}

%
Given $i,j\in\N$, we use $[i,j]$ for the set of natural numbers $h$
such that $i\leq h\leq j$, $[i,j)$ for the set of natural numbers $h$
such that $i\leq h< j$, and by $[i,\infty]$ the set of natural numbers
$h$ such that $h\geq i$.

We fix a \emph{finite} set $\AP$ of atomic propositions.
A \emph{trace} is an infinite word over $2^{\AP}$. A \emph{pointed
  trace} is a pair $(\pi,i)$ consisting of a trace $\pi$ and a
position $i\in\N$ along $\pi$.

For a word $w$ over some alphabet $\Sigma$, $|w|$ is the length of $w$
($|w|=\infty$ if $w$ is infinite), for each $0\leq i<|w|$, $w(i)$ is
the $(i+1)^{th}$ symbol of $w$, and $w^{i}$ is the suffix of $w$ from
position $i$, i.e., the word $w(i)w(i+1)\ldots$
\details{Given $n\geq 1$ and an $n$-tuple
  $\vect{w}=(w_1,\ldots,w_n)$ of infinite words over
  $\Sigma$, we identify $\vect{w}$ with the infinite word
  over $\Sigma^{n}$ whose $(i+1)^{th}$ symbol, for each $i\geq 0$, is
  given by $(w_1(i),\ldots,w_n(i))$.  }

Given $n,h\in\N$ and integer constants $c>1$, $\Tower_c(h,n)$
denotes a tower of exponentials of base $c$, height $h$, and argument
$n$: $\Tower_c(0,n)=n$ and $\Tower_c(h+1,n)=c^{\Tower_c(h,n)}$.
For each $h\in\N$, we denote by $h$-\EXPSPACE\ the class of languages
decided by deterministic Turing machines bounded in space by functions
of $n$ in $O(\Tower_c(h,n^d))$ for some integer constants $c> 1$ and
$d\geq 1$.
Note that $0$-\EXPSPACE\ coincides with \PSPACE.

 \subsection{Linear-time Temporal Logic (\LTL)}

 We recall syntax and semantics of $\LTL$~\cite{Pnueli77}.
 Formulas $\theta$ of $\LTL$ over the set $\AP$ of atomic propositions
 are defined as follows:
\[
\theta ::=   p  \ | \ \neg \theta \ | \ \theta \wedge \theta \ | \ \Next \theta   | \ \theta \Until \theta
\]
where $p\in \AP$ and  $\Next$ and $\Until$ are the
``next'' and ``until'' temporal modalities respectively.
The logic  is interpreted over pointed traces $(\pi,i)$. The satisfaction relation $(\pi,i)\models \theta$, meaning that formula $\theta$
holds at position $i$ along $\pi$, is inductively defined as follows (we omit the semantics for the Boolean connectives which standard):
\[ \begin{array}{ll}
     (\pi, i) \models  p  &  \Leftrightarrow  p \in \pi(i)\\
     (\pi,i)\models  \Next\theta &  \Leftrightarrow (\pi,i+1)\models  \theta\\
     (\pi, i) \models  \theta_1\Until \theta_2  &
  \Leftrightarrow  \text{for some $j\geq i$}: (\pi, j)
  \models  \theta_2
  \text{ and } \\
  & \phantom{\Leftrightarrow}\,\, (\pi, k) \models   \theta_1 \text{ for all }i\leq k<j
\end{array} \]
A trace $\pi$ is a model of $\theta$, written $\pi\models \theta$, if
$(\pi,0)\models \theta$.


\subsection{Linear-time Hyper Specifications}\label{sec:HyperSpecifications}

In this section, we consider an abstract notion of linear-time hyper
specifications which are interpreted over sets of traces.
For the rest of the discussion, we fix an ordered set $\Var$ of trace
variables.

A \emph{pointed trace assignment $\Pi$} is a partial mapping over
$\Var$, assigning to each trace variable $x$ in its domain $\Dom(\Pi)$
a pointed trace.
The assignment $\Pi$ is initial if for each $x\in \Dom(\Pi)$, $\Pi(x)$
is of the form $(\pi,0)$ for some trace $\pi$.
For a trace variable $x\in \Var$ and a pointed trace $(\pi,i)$, we
denote by $\Pi[x\mapsto (\pi,i)]$ the pointed trace assignment having
domain $\Dom(\Pi)\cup \{x\}$ that behaves as $\Pi$ on the variables in
$\Dom(\Pi)\setminus \{x\}$ and assigns to $x$ the pointed trace
$(\pi,i)$.

A \emph{multi-trace specification} $\S(x_1,\ldots,x_n)$ is a
specification (in some formalism) parameterized by a subset
$\{x_1,\ldots,x_n\}$ of $\Var$ whose semantics is represented by a set
$\Upsilon$ of pointed trace assignments with domain
$\{x_1,\ldots,x_n\}$.
Depending on the given formalism, one can restrict to consider only
\emph{initial} pointed trace assignments.
We write $\Pi\models S(x_1,\ldots,x_n)$ for the trace assignments
$\Pi$ in $\Upsilon$.

Given a class $\C$ of multi-trace specifications, linear-time hyper
expressions $\xi$ over $\C$ are defined as follows:
\[
\xi ::=    \exists x.  \xi \ | \ \forall x.  \xi \ | \ S(x_1,\ldots,x_n)\]
where $x,x_1,\ldots,x_n\in \Var$, $S(x_1,\ldots,x_n)$ is a multi-trace
specification in the class $\C$, $\exists x$ is the \emph{hyper}
existential trace quantifier for variable $x$, and $\forall x$ the
hyper universal trace quantifier for $x$.
Informally, the expression $\exists x. \xi$ requires that for some
trace $\pi$ in the given set of traces, $\xi$ holds when $x$ is mapped
to $(\pi,0)$, while $\forall x. \xi$ requires that all traces $\pi$,
$\xi$ holds when $x$ is mapped to $(\pi,0)$.
We say that an expression $\xi$ is a \emph{sentence} if every variable
$x_i$ in the multi-trace specification $S(x_1,\ldots,x_n)$ of $\xi$ is
in the scope of a quantifier for the trace variable $x_i$, and
distinct occurrences of quantifiers are associated with distinct trace
variables.
The \emph{quantifier alternation depth} of $\xi$ is the number of
switches between $\exists$ and $\forall$ quantifiers in the quantifier
prefix of $\xi$.

For instance, $\HLTL$ sentences~\cite{ClarksonFKMRS14} are linear-time
hyper sentences over the class of multi-trace specifications obtained
by $\LTL$ formulas by replacing atomic propositions $p$ with
relativized versions $\Rel{p}{x}$, where $x\in\Var$.
Intuitively, $\Rel{p}{x}$ asserts that $p$ holds at the pointed trace
assigned to $x$.


Given a linear-time expression $\xi$ with multi-trace specification
$S(x_1,\ldots,x_n)$, a set $\Lang$ of traces, and an initial pointed
trace assignment $\Pi$ such that $\Dom(\Pi)$ contains the variables in
$\{x_1,\ldots,x_n\}$ which are not in the scope of a quantifier, and
the traces referenced by $\Pi$ are in $\Lang$, the satisfaction
relation $(\Lang,\Pi)\models \xi$ is inductively defined as follows:
  \[ \begin{array}{ll}
  (\Lang,\Pi) \models  \exists x. \xi  &  \Leftrightarrow \text{ for some trace } \pi\in\Lang:\, \\
  & \phantom{\Leftrightarrow}\,\,\,\, (\Lang,\Pi[x\mapsto (\pi,0)]) \models  \xi \\
    (\Lang,\Pi) \models  \forall x. \xi  &  \Leftrightarrow \text{ for each trace } \pi\in\Lang:\, \\
    & \phantom{\Leftrightarrow}\,\,\,\, (\Lang,\Pi[x\mapsto (\pi,0)]) \models  \xi \\
(\Lang,\Pi) \models S(x_1,\ldots,x_n)  &  \Leftrightarrow \Pi\models S(x_1,\ldots,x_n)
\end{array}
\]
If $\xi$ is a sentence, we write $\Lang\models \xi$ to mean that
$(\Lang,\Pi_\emptyset)\models \xi$, where $\Pi_\emptyset$ is the empty
assignment.

\subsection{Kripke Structures and Asynchronous Word Automata}

\noindent \textbf{Kripke structures.}
A \emph{Kripke structure $($over
  $\AP$$)$} is a tuple $\Ku=\tpl{S,S_0,E,V}$, where
$S$ is a set of states, $S_0\subseteq
S$ is the set of initial states, $E\subseteq S\times
S$ is a transition relation which is total in the first argument (\ie
for each $s\in S$ there is a $t\in S$ with $(s,t)\in E$), and $V:S
\rightarrow 2^{\AP}$ is an
\emph{$\AP$-valuation} assigning to each state
$s$ the set of propositions 
holding at $s$.
The Kripke structure $\Ku$ is finite if $S$ is finite.

A \emph{path} $\nu= t_0,t_1,\ldots$ of
$\Ku$ is an infinite word over $S$ such that $t_0\in
S_0$ is an initial state and for all $i\geq 0$, $(t_{i},t_{i+1})\in
E$.
The path $\nu= t_0,t_1,\ldots$ induces the trace
$V(t_0)V(t_1)\ldots$.
A finite path of $\Ku$ is a non-empty finite infix of some path of
$\Ku$. 
A \emph{trace} of $\Ku$ is a trace induced by some path of $\Ku$.
We denote by $\Lang(\Ku)$ the set of traces of
$\Ku$.
We also consider \emph{fair finite Kripke structures}
$(\Ku,F)$, that is, finite Kripke structures
$\Ku$ equipped with a subset $F$ of $\Ku$-states.
A path $\nu$ of $\Ku$ is \emph{$F$-fair} if
$\nu$ visits infinitely many times states in
$F$.
We denote by $\Lang(\Ku,F)$ the set of traces of
$\Ku$ associated with the $F$-fair paths of
$\Ku$.
We consider the following decision problems for a given class
$\C$ of multi-trace specifications:
\begin{itemize}
\item \emph{Model checking problem}: checking for a given finite
  Kripke structure $\Ku$ and a linear-time hyper sentence
  $\xi$ over $\C$, whether $\Lang(\Ku)\models
  \xi$ (we also write $\Ku\models \xi$).
\item \emph{Fair model checking problem}: checking for a given fair
  finite Kripke structure
  $(\Ku,F)$ and a linear-time hyper sentence $\xi$ over
  $\C$, whether $\Lang(\Ku,F)\models \xi$.
\end{itemize}

Note that model checking reduces to fair model checking for the
special case where $F$ coincides with the set of
$\Ku$-states.\vspace{0.2cm}

\noindent \textbf{Labeled Trees.}
A tree $T$ is a prefix closed subset of $\N^{*}$.
Elements of $T$ are called nodes and the empty word $\varepsilon$ is
the root of $T$.
For $x\in T$, a child of $x$ in $T$ is a node of the form $x\cdot n$
for some $n\in \N$.
A path of $T$ is a maximal sequence $\pi$ of nodes such that
$\pi(0)=\varepsilon$ and $\pi(i)$ is a child in $T$ of $\pi(i-1)$ for
all $0<i<|\pi|$.
For an alphabet $\Sigma$, a $\Sigma$-labeled tree is a pair
$\tpl{T, \Lab}$ consisting of a tree and a labelling
$\Lab:T \Into \Sigma$ assigning to each node in $T$ a symbol in
$\Sigma$. \vspace{0.2cm}

\noindent\textbf{Asynchronous Word Automata.}
We consider a variant of the framework of alternating asynchronous
word automata introduced in~\cite{GutsfeldOO21}, a class of
finite-state automata for the asynchronous traversal of multiple
infinite words.
Given a set $X$, $\B^{+}(X)$ denotes the set of \emph{positive}
Boolean formulas over $X$, that is, Boolean formulas built from
elements in $X$ using $\vee$ and $\wedge$ (we also allow the formulas
$\true$ and $\false$).
\details{A subset $Y$ of $X$ \emph{satisfies} $\theta\in\B_+(X)$ iff
  the truth assignment that assigns $\true$ to the elements in $Y$ and
  $\false$ to the elements of $X\setminus Y$ satisfies $\theta$.
}
Let $n\geq 1$.
A B\"{u}chi $\NAAWA$ over a finite alphabet $\Sigma$ is a tuple
$\Au=\tpl{\Sigma,q_0,Q,\rho,F}$, where $Q$ is a finite set of
(control) states, $q_0\in Q$ is the initial state,
$\rho:Q\times \Sigma^{n}\rightarrow \B^{+}(Q\times [1,n])$ is the
transition function, and $F\subseteq Q$ is a set of accepting states.
Intuitively, an $\NAAWA$ has access to $n$ infinite input words over
$\Sigma$ and at each step, it activates multiple copies.
For each copy, there is exactly one input word whose current input
symbol is consumed, so the reading head of such word moves one
position to the right.

In particular, the target of a move of $\Au$ is encoded by a pair
$(q,i)\in A\times [1,n]$, where $q$ indicates the target state while
the direction $i$ indicates on which input word to progress.

Formally, a run of $\Au$ over an $n$-tuple
$\vect{w}=(w_1,\ldots,w_n)$ of infinite words over $\Sigma$
is a $(Q\times \N^{n})$-labeled tree $r=\tpl{T_r,\Lab_r}$, where each
node of $T_r$ labelled by $(q,\wp)$ with $\wp=(i_1,\ldots,i_n)$
describes a copy of the automaton that is in state $q$ and reads the
$(i_h+1)^{th}$ symbol of the input word $w_h$ for each $h\in[1,n]$.
Moreover, we require that
\begin{compactitem}
\item $r(\varepsilon)=(q_0,(0,\ldots,0))$, that is, initially, the
  automaton is in state $q_0$ reading the first position of each input
  word);
\item for each $\tau\in T_r$ with $\Lab_r(\tau)=(q,(i_1\ldots,i_n))$,
  there is a (possibly empty) set
  $\{(q_1,d_1),\ldots,(q_k,d_k)\}\subseteq Q\times [1,n]$ for some
  $k\geq 0$ satisfying $\delta(q,(w_1(i_1),\ldots,w_n(i_n)))$ such
  that $\tau$ has $k$ children $\tau_1,\ldots,\tau_k$ and
  $\Lab_r(\tau_j) = (q_j,(i_1,\ldots, i_{d_j}+1,\ldots,i_n))$ for all
  $1\leq j\leq k$.
\end{compactitem}  
The run $r$ is accepting if each infinite path $\nu$ visits infinitely
often nodes labeled by some accepting state in $F$.
We denote by $\Lang(\Au)$ the set of $n$-tuples $\vect{w}$
of infinite words over $\Sigma$ such that there is an accepting run of
$\Au$ over $\vect{w}$.

For each $k\geq 1$, we also consider \emph{$k$-synchronous} B\"{u}chi
$\NAAWA$~\cite{GutsfeldOO21}, which is a B\"{u}chi $\NAAWA$ such that
for each run $r$ and for each node of $r$ with label $(q,\wp)$, the
position vector $\wp=(i_1,\ldots,i_n)$ satisfies
$|i_\ell - i_{\ell'}|\leq k$ for all $\ell,\ell'\in [1,k]$.
Intuitively, a $k$-synchronous $\NAAWA$ can never be ahead more than
$k$ steps in one direction with respect to the others.
Note that $\AAWA$ over $2^{\AP}$ can be seen as multi-trace
specifications.
It is known~\cite{GutsfeldOO21} that model checking against
linear-time hyper sentences over B\"{u}chi $\AAWA$ is undecidable, and
the problem becomes decidable when one restricts to consider
$k$-synchronous B\"{u}chi $\AAWA$.
In particular, the following holds.

\begin{proposition}[\cite{GutsfeldOO21}]
  \label{prop:KsynchronousAWA}
  Let $d\in \N$.
  The (fair) model checking problem against linear-time hyper
  sentences of quantifier alternation depth $d$ over the class of
  $k$-synchronous B\"{u}chi $\NAAWA$ over $2^{\AP}$ ($k,n,\AP$ being
  input parameters of the problem instances) is
  $(d+1)$-\EXPSPACE-complete, and for a fixed formula, it is
  $(d-1)$-\EXPSPACE-complete for $d>0$ and \NLOGSPACE-complete
  otherwise.
\end{proposition}


\section{Stuttering $\HLTL$}\label{sec:StutteringHLTL}

In this section we introduce an asynchronous extension of $\HLTL$
that we call \emph{stuttering $\HLTL$} ($\SHLTL$ for short).
Stuttering $\HLTL$ is obtained by exploiting relativized versions of
the temporal modalities with respect to finite sets $\Gamma$ of $\LTL$
formulas.
Intuitively, these modalities are evaluated along sub-traces of the
given traces which are obtained by removing ``redundant'' positions
with respect to the pointwise evaluation of the $\LTL$ formulas in
$\Gamma$.
The rest of this section is organized as follows.
In Subsection~\ref{sec:RelativizedStuttering} we introduce a
generalization of the classical notion of stuttering.
Then, in Subsection~\ref{sec:SyntaxSemanticsHLTLS} we define the
syntax and semantics of $\SHLTL$ and provide some examples of
specifications in this novel logic.
Finally, we investigate the model checking problem against $\SHLTL$.
In Subsection~\ref{sec:UndecidabilitySHLTL}, we show that the problem
is in general undecidable, and in
Subsection~\ref{sec:DecidableSHLTLFragments}, we identify a meaningful
fragment of $\SHLTL$ for which model checking is shown to be
decidable.

\subsection{\LTL-Relativized Stuttering}\label{sec:RelativizedStuttering}

Classically, a trace is stutter-free if there are no consecutive
positions having the same propositional valuation unless the valuation
is repeated ad-infinitum.
We can associate to each trace a unique stutter-free trace by removing
``redundant'' positions.
In this subsection, we generalize these notions with respect to the
pointwise evaluation of a finite set of $\LTL$ formulas.

\begin{definition}[$\LTL$ stutter factorization]
  Let $\Gamma$ be a finite set of $\LTL$ formulas and $\pi$ a
  trace.
  The \emph{$\Gamma$-stutter factorization of $\pi$} is the unique
  increasing sequence of positions $\{i_k\}_{k\in[0,m_{\infty}]}$ for
  some $m_{\infty}\in \N\cup\{\infty\}$ such that the following holds
  for all $j <m_{\infty}$:
  \begin{compactitem}
  \item $i_0=0$ and $i_j<i_{j+1}$;
  \item for each $\theta\in \Gamma$, the truth value of $\theta$ along
    the segment $[i_j,i_{j+1})$ does not change, \ie for all
    $h,k\in [i_j,i_{j+1})$, $(\pi,h)\models \theta$ iff
    $(\pi,k)\models \theta$, and the same holds for the infinite segment
    $[m_{\infty}, \infty]$ in case $m_{\infty}\neq \infty$;
  \item the truth value of some formula in $\Gamma$ changes along
    adjacent segments, \ie for some $\theta\in \Gamma$ (depending on
    $j$), $(\pi,i_j)\models \theta$ iff $(\pi,i_{j+1})\not\models \theta$.
\end{compactitem}
\end{definition}

Thus, the $\Gamma$-stutter factorization
$\{i_k\}_{k\in[0,m_{\infty}]}$ of $\pi$ partitions the trace in
adjacent non-empty segments such that the valuation of formulas in
$\Gamma$ does not change within a segment, and changes in moving from
a segment to the adjacent ones.
This factorization induces in a natural way a trace obtained by
selecting the first positions of the finite segments and all the
positions of the unique infinite segment, if any.
Formally, the \emph{$\Gamma$-stutter trace of $\pi$}, denoted by
$\stfr_{\Gamma}(\pi)$, is defined as follows:
\begin{itemize}
\item $\stfr_{\Gamma}(\pi)\DefinedAs\pi(i_0)\pi(i_1)\ldots $ if $m_{\infty}=\infty$;
\item $\stfr_{\Gamma}(\pi)\DefinedAs\pi(i_0)\pi(i_1)\ldots \pi(i_{m_{\infty}-1})\cdot \pi^{i_{m_{\infty}}}$ if $m_{\infty}\neq\infty$.
\end{itemize}

As an example, assume that $\AP =\{p,q,r\}$ and let
$\Gamma= \{p\Until q\}$.
Given $h,k\geq 1$, let $\pi_{h,k}$ be the trace
$\pi_{h,k} = p^{h} q^{k} r^{\omega}$.
All these traces have the same $\Gamma$-stutter trace which is given
by $p r^{\omega}$.
This is because for $p\Until q$ is true for $p^{h'}q^{k'} r^{\omega}$
and $q^{k'} r^{\omega}$ (for all $h'$ and $k'$) and false for $r^{\omega}$.
Therefore, there are two change points which are the initial position
(with valuation $p$) and the first $r$ position.

We say that a trace $\pi$ is \emph{$\Gamma$-stutter free} if it
coincides with its $\Gamma$-stutter trace,
\ie~$\stfr_{\Gamma}(\pi)= \pi$.
Note that if $\Gamma=\emptyset$, each trace is $\emptyset$-stutter
free, i.e.~$\stfr_{\emptyset}(\pi)= \pi$.
\details{ The notion of $\Gamma$-stutter trace of a trace $\pi$
  induces an equivalence relation over the set of traces: two traces
  $\pi$ and $\pi'$ are \emph{$\Gamma$-stutter equivalent} if
  $\stfr_{\Gamma}(\pi)= \stfr_{\Gamma}(\pi')$.
  If $\Gamma=\emptyset$, then $\emptyset$-stutter equivalence classes
  are singletons, and each trace is $\emptyset$-stutter free.
  Note that the classical notion of stutter equivalence corresponds to
  the $\AP$-stutter equivalence.
  In this case, we have that the $\AP$-stutter trace of a trace $\pi$
  is $\AP$-stutter equivalent to $\pi$.
  This property evidently does not hold for arbitrary finite sets of
  $\LTL$ formulas.
  For instance, let us consider the trace $\pi= p r q r^{\omega}$.
  We have that $\stfr_{\{p\Until q\}}(\pi)= p q r^{\omega}$.
  On the other hand, the $\{p\Until q\}$-stutter free trace of
  $p q r^{\omega}$ is $p r^{\omega}$.
}

For each finite set $\Gamma$ of $\LTL$ formulas, we define the
successor function $\SUCC_\Gamma$ as follows.
The function maps a pointed trace $(\pi,i)$ to the trace $(\pi,\ell)$
where $\ell$ is the first position of the segment in the
$\Gamma$-stutter factorization of $\pi$ following the $i$-segment (if
the $i$-segment is not the last one).
If the segment is the last one then $\ell$ is $i+1$.
Formally,

\begin{definition}[Relativized Successor]
  Let $\Gamma$ be a finite set of $\LTL$ formulas, $\pi$ a trace with
  $\Gamma$-stutter factorization $\{i_k\}_{k\in[0,m_{\infty}]}$, and
  $i\geq 0$.
  The \emph{$\Gamma$-successor of the pointed trace $(\pi,i)$},
  denoted by $\SUCC_\Gamma(\pi,i)$, is the trace $(\pi,\ell)$ where
  position $\ell$ is defined as follows: if there is $j<m_{\infty}$
  such that $i\in [i_j,i_{j+1}-1]$, then $\ell= i_{j+1}$; otherwise
  (note that in this case $m_{\infty}\neq \infty$ and
  $i\geq i_{m_{\infty}}$), $\ell=i+1$.
\end{definition}

\subsection{Syntax and Semantics of Stuttering $\HLTL$ ($\SHLTL$)}%
\label{sec:SyntaxSemanticsHLTLS}

$\SHLTL$ formulas over the given finite set $\AP$ of atomic
propositions and finite set $\Var$ of trace variables are linear-time
hyper expressions over multi-trace specifications $\psi$, called
\emph{$\SHLTL$ quantifier-free formulas}, where $\psi$ is defined by
the following syntax:
\[
  \psi ::=    \top \ | \  \Rel{p}{x}  \ | \ \neg \psi \ | \ \psi \wedge \psi \ | \ \Next_\Gamma \psi  \ | \ \psi \Until_\Gamma \psi
\]
where $p\in \AP$, $x\in \Var$, $\Gamma$ is a finite set of $\LTL$
formulas over $\AP$, and $\Next_\Gamma$ and $\Until_\Gamma$ are the
stutter-relativized versions of the $\LTL$ temporal modalities.

When $\Gamma$ is empty, we omit the subscript $\Gamma$ in the temporal
modalities.
Informally, $\Rel{p}{x}$ asserts that $p$ holds at the pointed trace
assigned to $x$, while the relativized temporal modalities
$\Next_\Gamma$ and $\Until_\Gamma$ are evaluated by a lockstepwise
traversal of the $\Gamma$-stutter traces associated with the currently
quantified traces.
We also exploit the standard logical connectives $\vee$ (disjunction)
and $\rightarrow$ (implication) as abbreviations, and the
\emph{relativized eventually} modality
$\Eventually_\Gamma \psi \DefinedAs \top \Until_\Gamma \psi$ and its dual
$\Always_\Gamma \psi \DefinedAs \neg \Eventually_\Gamma\neg \psi$
(\emph{relativized always}).
The size $|\xi|$ of a $\SHLTL$ (quantifier-free) formula $\xi$ is the
number of distinct sub-formulas of $\xi$ plus the number of distinct
sub-formulas of those $\LTL$ formulas occurring in the subscripts of
the temporal modalities.

For each finite set $\Gamma$ of $\LTL$ formulas, we denote by
$\SHLTL[\Gamma]$ the syntactical fragment of $\SHLTL$ where the
subscript of each temporal modality is $\Gamma$.
Note that standard $\HLTL$ corresponds to the fragment
$\SHLTL[\emptyset]$.
In the following, for each $\SHLTL$ formula $\varphi$, we denote by
$\HLTL(\varphi)$ the $\HLTL$ formula obtained from $\varphi$ by
replacing each relativized temporal modality in $\varphi$ with its
$\emptyset$-relativized version.
\vspace{0.2cm}

\noindent \textbf{Semantics of  $\SHLTL$ Quantifier-free Formulas.}
Given a finite set $\Gamma$ of $\LTL$ formulas, we extend in a natural
way the relativized successor function $\SUCC_\Gamma$ to pointed
trace assignments $\Pi$ as follows: the \emph{$\Gamma$-successor
  $\SUCC_\Gamma(\Pi)$ of $\Pi$} is the pointed trace assignment with
domain $\Dom(\Pi)$ associating to each $x\in\Dom(\Pi)$ the
$\Gamma$-successor $\SUCC_\Gamma(\Pi(x))$ of the pointed trace
$\Pi(x)$.
For each $j\in \N$, we use $\SUCC^{\,j}_\Gamma$ for the function
obtained by $j$ applications of the function $\SUCC_\Gamma$.

Given a $\SHLTL$ quantifier-free formula $\psi$ and a pointed trace
assignment $\Pi$ such that $\Dom(\Pi)$ contains the trace variables
occurring in $\psi$, the satisfaction relation $\Pi\models \psi$ is
inductively defined as follows (we omit the semantics of the Boolean
connectives which is standard):
  \[ \begin{array}{ll}
 \Pi  \models \Rel{p}{x}  &  \Leftrightarrow  \Pi(x)=(\pi,i) \text{ and }p\in \pi(i)\\
   \Pi  \models  \Next_\Gamma\psi &  \Leftrightarrow  \SUCC_\Gamma(\Pi) \models  \psi\\
   \Pi  \models  \psi_1\Until_\Gamma \psi_2  &
  \Leftrightarrow  \text{for some }i\geq 0:\,   \SUCC^{\,i}_\Gamma(\Pi) \models  \psi_2 \text{ and }\\
   & \phantom{\Leftrightarrow}  \SUCC^{\,k}_\Gamma(\Pi) \models  \psi_1 \text{ for all } 0\leq k<i
\end{array} \]
In the following, given a set $\Gamma$ of $\LTL$ formulas, we also
consider the model checking problem against the fragment
$\SHLTL[\Gamma]$ of $\SHLTL$.
For this fragment, by the semantics of $\SHLTL$, we deduce the
following fact, where for a set $\Lang$ of traces,
$\stfr_\Gamma(\Lang)$ denotes the set of $\Gamma$-stutter traces over
the traces in $\Lang$,
\ie $\stfr_\Gamma(\Lang)\DefinedAs\{\stfr_\Gamma(\pi)\mid \pi\in \Lang\}$.

\begin{remark}\label{remark:connection}
  A set $\Lang$ of traces is a model of a $\SHLTL[\Gamma]$ sentence
  $\varphi$ if and only if $\stfr_\Gamma(\Lang)$ is a model of the
  $\HLTL$ sentence $\HLTL(\varphi)$.
\end{remark}

Let $\SLTL$ be the extension of $\LTL$ obtained by adding the
stutter-relativized versions of the $\LTL$ temporal modalities.
Note that $\SLTL$ formulas correspond to \emph{one-variable} $\SHLTL$
quantifier-free formulas.
We can show that $\SLTL$ has the same expressiveness as $\LTL$, as
established by the following Proposition~\ref{prop:FromSLTLtoLTL}
(missing proofs of all the claims in this paper can be found
in the Appendix).
On the other hand, $\SHLTL$ quantifier-free formulas are in general
more expressive than $\HLTL$ quantifier-free formulas.
In particular, while model checking $\HLTL$ is known to be
decidable~\cite{ClarksonFKMRS14}, model checking the alternation-free
fragment of $\SHLTL$ is already undecidable (see
Subsection~\ref{sec:UndecidabilitySHLTL}).

\newcounter{prop-FromSLTLtoLTL}
\setcounter{prop-FromSLTLtoLTL}{\value{proposition}}
\newcounter{sec-FromSLTLtoLTL}
\setcounter{sec-FromSLTLtoLTL}{\value{section}}

\begin{proposition}\label{prop:FromSLTLtoLTL}
  Given a $\SLTL$ formula, one can construct in polynomial time an
  equivalent $\LTL$ formula.
\end{proposition}

We now show that $\SHLTL$ is strictly less expressive than the
fixpoint calculus $\HU$ introduced in~\cite{GutsfeldOO21}.
Indeed, $\HU$ cannot be embedded into $\SHLTL$ since for singleton
trace sets, $\HU$ characterizes the class of $\omega$-regular
languages, while $\SHLTL$ corresponds to $\LTL$, which consequently,
captures only a strict subclass of $\omega$-regular languages.
Moreover, by the following result and the fact that parity $\AAWA$ are
equivalent to $\HU$ quantifier-free formulas~\cite{GutsfeldOO21}, we
obtain that $\SHLTL$ is subsumed by $\HU$.

\newcounter{prop-FromStutteringHLTLtoAAWA}
\setcounter{prop-FromStutteringHLTLtoAAWA}{\value{proposition}}
\newcounter{sec-FromStutteringHLTLtoAAWA}
\setcounter{sec-FromStutteringHLTLtoAAWA}{\value{section}}

\begin{proposition}\label{prop:FromStutteringHLTLtoAAWA}
  Given a $\SHLTL$ quantifier-free formula $\psi$ with trace variables
  $x_1,\ldots,x_n$, one can build in polynomial time a B\"{u}chi
  $\NAAWA$ $\Au_\psi$ such that $\Lang(\Au_\psi)$ is the set of
  $n$-tuples $(\pi_1,\ldots,\pi_n)$ of traces so that
  $(\{x_1\mapsto (\pi_1,0),\ldots,x_1\mapsto (\pi_n,0)\})\models
  \psi$.
\end{proposition}

\begin{proof}
  By exploiting the dual $\Release_\Gamma$ (\emph{relativized
    release}) of the until modality $\Until_\Gamma$, we can assume
  without loss of generality that $\psi$ is in \emph{negation normal
    form}, so negation is applied only to relativized atomic
  propositions.
  Given a finite set $\Gamma$ of $\LTL$ formulas, let $\xi_{\Gamma}$ be
  the following $\LTL$ formula
  \[
    \xi_\Gamma = \displaystyle{\bigwedge_{\xi\in \Gamma} \Always(\xi \leftrightarrow \Next \xi) \vee \bigvee_{\xi\in\Gamma} (\xi \leftrightarrow \neg \Next\xi)}
  \]
  \noindent The $\LTL$ formula $\xi_\Gamma$ has as models the traces
  $\pi$ such that the first segment in the factorization of $\pi$ is
  either infinite or has length $1$.
  For each $i\in [1,n]$, we can easily construct in linear time (in
  the number of distinct sub-formulas in $\Gamma$) a B\"{u}chi
  $\NAAWA$ $\Au_{\Gamma,i}$ 
  accepting the $n$-tuples $(\pi_1,\ldots,\pi_n)$ of traces so that
  the $i^{th}$ component $\pi_i$ is a model of
  $\xi_\Gamma$.
  Similarly, we can also define $\overline{\Au}_{\Gamma,i}$ accepting
  the $n$-tuples $(\pi_1,\ldots,\pi_n)$ of traces so that the $i^{th}$
  component $\pi_i$ is not a model of $\xi_\Gamma$.  
 
  Let $\Upsilon$ be the set of subscripts $\Gamma$ occurring in the
  temporal modalities of $\psi$.
  Then by exploiting the automata $\Au_{\Gamma,i}$ and
  $\overline{\Au}_{\Gamma,i}$ where $\Gamma\in\Upsilon$ and
  $i\in [1,n]$, we construct a B\"{u}chi $\NAAWA$ $\Au_\psi$
  satisfying Proposition~\ref{prop:FromStutteringHLTLtoAAWA} as
  follows.
  Given an input multi-trace $(\pi_1,\ldots,\pi_n)$, the behaviour of
  the automaton $\Au_\psi$ is subdivided in phases.
  At the beginning of each phase with current position vector
  $\wp=(j_1,\ldots,j_n)$, $\Au_\psi$ keeps track in its state of the
  currently processed sub-formula $\theta$ of $\psi$.
  By the transition function, $\theta$ is processed in accordance with
  the `local' characterization of the semantics of the Boolean
  connectives and the relativized temporal modalities.
  Whenever $\theta$ is of the form $\theta_1\Until_\Gamma\theta_2$ or
  $\theta_1\Release_\Gamma\theta_2$, or $\theta$ is argument of a
  sub-formula of the form $\Next_\Gamma\theta$, and $\Au_\psi$ has to
  check that $\theta$ holds at the position vector
  $(\SUCC_\Gamma(\pi_1,j_1),\ldots,\SUCC_\Gamma(\pi_1,j_1))$,
  $\Au_\psi$ moves along the directions $1,\ldots,n$ in turns.
  During the movement along direction $i\in [1,n]$, the automaton is
  in state $(\theta,i,\Gamma)$ and guesses that either
  \begin{inparaenum}[(i)]
  \item the next input position is in the current segment of the
    $\Gamma$-factorization of $\pi_i$ and this segment is not the last
    one, or
  \item the previous condition does not hold, hence, the next input
    position corresponds to $\SUCC_\Gamma(\pi_1,j_1)$.
  \end{inparaenum}
  In the first (resp., second case) case, it activates in parallel a
  copy of the auxiliary automaton $\overline{\Au}_{\Gamma,i}$ (resp.,
  $\Au_{\Gamma,i}$) for checking that the guess is correct and moves
  one position to the right along $\pi_i$.
  Moreover, in the first case, $\Au_\psi$ remains in state
  $(\theta,i,\Gamma)$, while in the second case, the automaton changes
  direction by moving to the state $(\theta,i+1,\Gamma)$ if $i<n$, and
  starts a new phase by moving at state $\theta$ otherwise.
\end{proof}

\noindent \textbf{Examples of Specifications.}
Stuttering $\HLTL$ can express relevant information-flow security
properties for asynchronous frameworks such as distributed systems or
cryptographic protocols.
These properties specify how information may propagate from input to
outputs by comparing distinct executions of a system possibly at
different points of time.
Assume that each user is classified either at a low security level,
representing public information, or at a high level, representing
secret information.
Moreover, let $LI$ be a set of propositions for describing inputs of
low users, $LO$ propositions that describe outputs of low users, and
$HI$ be a set of propositions for representing inputs of high users.
As a first example, we consider the asynchronous variant of the
\emph{noninterference} property, as defined by Goguen and
Meseguer~\cite{goguen1982security}, asserting that the observations of
low users do not change when all high inputs are removed.
In an asynchronous setting, a user cannot infer that a transition
occurred if consecutive observations remain unchanged.
In other terms, steps observed by a user do not correspond to the
same number of steps in different executions of the system.
Thus, since a low user can only observe the low output propositions,
we require that for each trace $\pi$, there is a trace $\pi'$ with no
high inputs such that the $LO$-stutter traces of $\pi$ and $\pi'$
coincide, that is, $\pi$ and $\pi'$ are indistinguishable to a low
user.
This can be expressed in $\SHLTL$ as follows, where proposition
$p_\emptyset$ denotes absence of high input.
\[
\forall x.\, \exists y.\,\Always \Rel{p_\emptyset}{y} \wedge \Always_{LO} \bigwedge_{p\in LO} (\Rel{p}{x} \leftrightarrow \Rel{p}{y})
\]
Assuming that the observations are not time-sensitive,
noninterference cannot in general be expressed in $\HLTL$ unless one
only considers systems where all the traces are $LO$-stutter free.
Another relevant example is \emph{generalized noninterference} as
formulated in~\cite{McLean96} which allows nondeterminism in the
low-observable behavior and requires for all system traces $\pi$ and
$\pi'$, the existence of an interleaved trace $\pi''$ whose high
inputs are the same as $\pi$ and whose low outputs are the same as
$\pi'$.
This property can be expressed in $\SHLTL$ as follows:
\[
\forall x.\,\forall y.\, \exists z.\, \Always_{HI} \bigwedge_{p\in HI} (\Rel{p}{y} \leftrightarrow \Rel{p}{z}) \wedge \Always_{LO} \bigwedge_{p\in LO} (\Rel{p}{y} \leftrightarrow \Rel{p}{z})
\]
Another classical security policy is \emph{observational determinism}
specifying that traces which have the same initial low inputs are
indistinguishable to a low user.
The following $\SHLTL$ formula captures observational determinism with
equivalence of traces up to stuttering as formulated
in~\cite{ZdancewicM03}.
\[
\forall x.\,\forall y.\,   \bigwedge_{p\in LI}  (\Rel{p}{x} \leftrightarrow \Rel{p}{y}) \rightarrow \Always_{LO} \bigwedge_{p\in LO} (\Rel{p}{x} \leftrightarrow \Rel{p}{y})
\]
Lastly, an interesting feature of $\SHLTL$ is the possibility of
combining asynchrony and synchrony constraints.
We illustrate this ability by considering an unbounded time
requirement which has application in the analysis of procedural
software:
``\emph{whenever a procedure $A$ is invoked, the procedure terminates,
  but there is no bound on the running time of $A$ that upper-bounds
  the duration of $A$ on all traces}''.
In other words, for every candidate bound $k$ there is a trace in
which $A$ is invoked and terminates with a longer duration.
We assume, crucially, that procedure $A$ can be activated at most once
along an execution, and we let $c_A$ characterize the call to $A$ and
$r_A$ the return.
This requirement can be expressed in $\SHLTL$ as follows.
%
\[
  \forall x.\,\exists y.\, \Eventually c_A[x] \rightarrow
  \begin{pmatrix}
    \begin{array}{@{}l@{}}
      \Eventually c_A[y] \wedge\Next_{\{\Eventually c_A\}}
      \begin{pmatrix}
        \begin{array}{@{}c@{}}
          \neg \Rel{r_A}{x} \wedge \neg \Rel{r_A}{y}\\
          \Until_{\{\Eventually c_A\}}\\
          \Rel{r_A}{x}\wedge \neg\Rel{r_A}{y}                   
        \end{array}
      \end{pmatrix}\end{array}\end{pmatrix}
\]
Essentially, we claim there is always a call to $A$ that runs for a
longer period of time than any candidate maximum duration.
Note that the occurrence of the relativized until
$\Until_{\{\Eventually c_A\}}$ in the previous formula can be
equivalently replaced by the standard until $\Until$.
We have used the relativized until since the previous formula is in
the fragment investigated in
Subsection~\ref{sec:DecidableSHLTLFragments} below which enjoys a
decidable model checking problem.

\subsection{Undecidability of Model Checking $\SHLTL$}
\label{sec:UndecidabilitySHLTL}

In this section, we establish the following negative result.

\begin{theorem}\label{theo:UndecidabilitySHLTL}
  The model checking problem for $\SHLTL$ is undecidable even for the
  $\SHLTL$ fragment where the quantifier alternation depth is $0$ and
  the stutter-relativized temporal modalities just use two sets of
  $\LTL$-formulas where one is empty and the other one consists of
  atomic propositions only.
\end{theorem}

Theorem~\ref{theo:UndecidabilitySHLTL} is proved by a reduction from
the Post's Correspondence Problem (PCP, for short)~\cite{HU79}.
We fix an instance $\Instance$ of PCP which is a tuple
 \[
 \Instance=\tpl{\tpl{u_1^{1},\ldots,u_n^{1}},\tpl{u_1^{2},\ldots,u_n^{2}}}
 \]
 where $n\geq 1$ and for each $1\leq i\leq n$, $u^{1}_i$ and $u^{2}_i$
 are non-empty finite words over a finite alphabet $\Sigma$.
 Let $[n]=\{1,\ldots,n\}$.
 A \emph{solution of} $\Instance$ is a non-empty sequence
 $i_1,i_2,\ldots,i_k$ of integers in $[n]$ such that
 $u^{1}_{i_1}\cdot u^{1}_{i_2}\cdot\ldots\cdot
 u^{1}_{i_k}=u^{2}_{i_1}\cdot u^{2}_{i_2}\cdot\ldots\cdot
 u^{2}_{i_k}$.
 PCP consists in checking for a given instance $\Instance$, whether
 $\Instance$ admits a solution. This problem is known to be
 undecidable~\cite{HU79}.\vspace{0.1cm}

 \noindent \textbf{Assumption.}
 We assume without loss of generality that each word $u^{\ell}_i$ of
 $\Instance$, where $i\in [n]$ and $\ell=1,2$, has length at least
 $2$.
 Indeed, if this assumption does not hold, we consider the instance
 $\Instance'$ of PCP obtained from $\Instance$ by replacing each word
 $u^{\ell}_i$ of the form $a_1\ldots a_n$ with the word
 $a_1 a_1\ldots a_n a_n$ (i.e., we duplicate each symbol occurring in
 $u_i^{\ell}$).
 Evidently $\Instance'$ has a solution if and only if $\Instance$ has
 a solution. \vspace{0.1cm}

 In order to encode the PCP instance $\Instance$ into an instance of
 the model checking problem for $\SHLTL$, we exploit the following set
 $\AP$ of atomic propositions, where $\#,p_1,\ldots,p_n,q_1,q_2$ are
 fresh symbols not in $\Sigma$.
   \[
   \AP\DefinedAs \Sigma\cup \{\#\}\cup \{p_1\ldots,p_n\}\cup \{q_1,q_2\}
   \]
   Intuitively, for each $i\in [n]$ and $\ell =1,2$, propositions
   $p_i$ and $q_\ell$ are exploited to mark each symbol of the word
   $u_i^{\ell}$ of the instance $\Instance$, while proposition $\#$ is
   used to mark only the last symbol of $u_i^{\ell}$.
   Thus, for the word $u_i^{\ell}$, we denote by
   $[u_i^{\ell},p_i,q_\ell]$ the finite word over $2^{\AP}$ of length
   $|u_i^{\ell}|$ obtained from $u_i^{\ell}$ by marking each symbol of
   $u_i^{\ell}$ with propositions $p_i$ and $q_\ell$ and,
   additionally, by marking the last symbol of $u_i^{\ell}$ with
   proposition $\#$.
   Formally, $[u_i^{\ell},p_i,q_\ell]$ is the finite word over
   $2^{\AP}$ having length $|u_i^{\ell}|$ such that for each
   $0\leq h<|u_i^{\ell}|$,
   $[u_i^{\ell},p_i,q_\ell](h)= \{u_i^{\ell}(h),p_i,q_\ell\}$ if
   $h<|u_i^{\ell}|-1$, and
   $[u_i^{\ell},p_i,q_\ell](h)= \{u_i^{\ell}(h),p_i,q_\ell,\#\}$
   otherwise.
   
   Given a non-empty sequence $i_1,i_2,\ldots,i_k$ of integers in
   $[n]$ and $\ell=1,2$, we encode the word
   $u^{\ell}_{i_1}\cdot u^{\ell}_{i_2}\cdot\ldots\cdot u^{\ell}_{i_k}$
   by the trace, denoted by $\pi^{\ell}_{i_1,\ldots,i_k}$, defined as:
   \[
   \pi^{\ell}_{i_1,\ldots,i_k} \DefinedAs \{\#\}\cdot [u_{i_1}^{\ell},p_{i_1},q_\ell] \cdot \ldots \cdot [u_{i_k}^{\ell},p_{i_k},q_\ell] \cdot \{\#\}^{\omega}
   \]
   Let $\Gamma$ be the set of atomic propositions given by
   $\Gamma=\{\#,p_1,\ldots,p_n\}$.
   We crucially observe that since each word of $\Instance$ has length
   at least $2$, the projection of the $\Gamma$-stutter trace
   $\stfr_{\Gamma}(\pi^{\ell}_{i_1,\ldots,i_k})$ of
   $\pi^{\ell}_{i_1,\ldots,i_k}$ over $\Gamma$ is given by
   \[
   \{\#\} \cdot \{p_{i_1}\} \cdot \{p_{i_1},\#\} \cdot\ldots   \cdot \{p_{i_k}\}\cdot \{p_{i_k},\#\}\cdot \{\#\}^{\omega}
   \]
   Hence, we obtain the following characterization of non-emptiness of
   the set of $\Instance$'s solutions, where a \emph{well-formed
     trace} is a trace of the form $\pi^{\ell}_{i_1,\ldots,i_k}$ for
   some non-empty sequence $i_1,i_2,\ldots,i_k$ of integers in $[n]$
   and $\ell=1,2$.

   \begin{proposition}\label{prop:CharcaterizeSolutions}
     $\Instance$ has some solution \emph{if and only} if there are two
     well-formed traces $\pi_1$ and $\pi_2$ satisfying the following
     conditions, where $\Gamma= \{\#,p_1,\ldots,p_n\}$:
   \begin{compactenum}
   \item for each $\ell=1,2$, $\pi_\ell$ does not contain occurrences
     of propositions $q_{3-\ell}$, i.e. for each $h\in\N$,
     $q_{3-\ell}\notin \pi_\ell(h)$;
   \item the projections of $\pi_1$ and $\pi_2$ over $\Sigma$
     coincide, i.e. for each $h\in\N$ and $p\in \Sigma$,
     $p\in \pi_1(h)$ iff $p\in \pi_2(h)$;
   \item the projections of $\stfr_{\Gamma}(\pi_1)$ and
     $\stfr_{\Gamma}(\pi_2)$ over $\Gamma$ coincide, i.e. for each
     $h\in\N$ and $p\in \Gamma$, $p\in \stfr_{\Gamma}(\pi_1)(h)$ iff
     $p\in \stfr_{\Gamma}(\pi_2)(h)$.
   \end{compactenum}
   \end{proposition}

   By exploiting Proposition~\ref{prop:CharcaterizeSolutions}, we
   construct a finite Kripke structure $\Ku_\Instance$ and a $\SHLTL$
   sentence $\varphi_\Instance$ over $\AP$ whose quantifier
   alternation depth is $0$ and whose temporal modalities are
   parameterized either by the empty set or by
   $\Gamma= \{\#,p_1,\ldots,p_n\}$ such that $\Instance$ has a
   solution if and only if $\Ku_\Instance\models
   \varphi_\Instance$.
   Note that Theorem~\ref{theo:UndecidabilitySHLTL} then follows
   directly by the undecidability of PCP.

   First, we easily deduce the following result concerning the
   construction of the Kripke structure $\Ku_\Instance$.

\newcounter{prop-KripkeBuild}
\setcounter{prop-KripkeBuild}{\value{proposition}}
\newcounter{sec-KripkeBuild}
\setcounter{sec-KripkeBuild}{\value{section}}

\begin{proposition}\label{prop:KripkeBuild}
  One can build in time polynomial in the size of $\Instance$ a finite
  Kripke structure $\Ku_\Instance$ over $\AP$ satisfying the following
  conditions:
   \begin{compactitem}
   \item the set of traces of $\Ku_\Instance$ contains the set of
     well-formed traces;
   \item each trace of $\Ku_\Instance$ having a suffix where $\#$
     always holds is a well-formed trace.
   \end{compactitem}
   \end{proposition}

   \noindent Finally, the $\SHLTL$ sentence $\varphi_\Instance$ is
   defined as follows, where $\Gamma= \{\#,p_1,\ldots,p_n\}$:
\[
\begin{array}{ll}
  \varphi_{\Instance} ::=  &  \exists x_1.\,\exists x_2.\,\, \Eventually  \Always (\#[x_1]\wedge \#[x_2])\, \wedge\, \vspace{0.1cm} \\
  & \Always  (\neg \Rel{q_2}{x_1} \wedge \neg \Rel{q_1}{x_2})\,\wedge\, \vspace{0.1cm}\\
  & \displaystyle{\bigwedge_{p\in \Sigma}}\Always  (\Rel{p}{x_1} \leftrightarrow \Rel{p}{x_2}]\,\wedge\,\displaystyle{\bigwedge_{p\in \Gamma}}\Always_{\Gamma} (\Rel{p}{x_1}] \leftrightarrow \Rel{p}{x_2})
\end{array}
\]
Assume that $\varphi_\Instance$ is interpreted over the Kripke
structure $\Ku_\Instance$ of Proposition~\ref{prop:KripkeBuild}.
Then, by Proposition~\ref{prop:KripkeBuild}, the first conjunct in the
body of $\varphi_\Instance$ ensures that the two traces $\pi_1$ and
$\pi_2$ of $\Ku_\Instance$ selected by the existential quantification
are well-formed traces.
Moreover, the other three conjuncts in the body of $\varphi_\Instance$
correspond to Conditions~(1)--(3) of
Proposition~\ref{prop:CharcaterizeSolutions} over the selected traces
$\pi_1$ and $\pi_2$.
Hence, $\Ku_\Instance \models \varphi_\Instance$ if and only if
there are two well-formed traces that satisfy Conditions~(1)--(3) of
Proposition~\ref{prop:CharcaterizeSolutions} if and only if
$\Instance$ admits a solution.
This concludes the proof of Theorem~\ref{theo:UndecidabilitySHLTL}.

\subsection{A Decidable Fragment of $\SHLTL$}
\label{sec:DecidableSHLTLFragments}

In the previous section, we have shown that the model checking problem
is undecidable for the $\SHLTL$ sentences whose relativized temporal
modalities exploit two distinct sets of \LTL\ formulas.
In this section, we establish that the use of a unique finite set
$\Gamma$ of $\LTL$ formulas as a subscript of the temporal modalities
in the given formula leads to a decidable model checking problem.
In particular, we consider the fragment of $\SHLTL$, we call
\emph{simple $\SHLTL$}, whose quantifier-free formulas $\psi$ satisfy
the following requirement: there exists a finite set $\Gamma$ of
$\LTL$ formulas (depending on $\psi$) such that $\psi$ is a Boolean
combination of quantifier-free formulas in $\SHLTL[\Gamma]$ and
\emph{one-variable} $\SHLTL$ quantifier-free formulas.

Simple $\SHLTL$ strictly subsumes $\HLTL$ and can express interesting
asynchronous security properties like asynchronous
noninterference~\cite{goguen1982security} and observational
determinism~\cite{ZdancewicM03}.
In particular, the $\SHLTL$ sentences at the end of
Subsection~\ref{sec:SyntaxSemanticsHLTLS} used for expressing
noninterference (but not generalized noninterference), observational
determinism, and the unbounded time procedural requirement are simple
$\SHLTL$ formulas.

We solve the (fair) model checking for simple $\SHLTL$ by a reduction
to $\HLTL$ model checking, which is known to be
decidable~\cite{ClarksonFKMRS14}.
Our reduction is exponential in the size of the given sentence.
As a preliminary step, we first show, by an adaptation of the standard
automata-theoretic approach for $\LTL$~\cite{VardiW94}, that the
problem for a simple $\SHLTL$ sentence $\varphi$ can be reduced in
exponential time to the fair model checking against a sentence in the
fragment $\SHLTL[\Gamma]$ for some set $\Gamma$ of \emph{atomic
  propositions} depending on $\varphi$.\vspace{0.1cm}

\noindent \textbf{Reduction to the Fragment $\SHLTL[\Gamma]$ with
  $\Gamma$ being Propositional.}
In order to prove the reduction from simple $\SHLTL$ to
$\SLTL[\Gamma]$ for propositional $\Gamma$ (formally expressed in
Theorem~\ref{theorem:FormSimpleToProposition} below), we need some
preliminary results.
Recall that a Nondeterministic B\"{u}chi Automaton ($\NBA$) is a tuple
$\Au =\tpl{\Sigma,Q,Q_0,\Delta,\Acc}$, where $\Sigma$ is a finite
alphabet, $Q$ is a finite set of states, $Q_0\subseteq Q$ is the set
of initial states, $\Delta\subseteq Q\times \Sigma \times Q$ is the
transition function, and $\Acc\subseteq Q$ is the set of
\emph{accepting} states. Given a infinite word $w$ over $\Sigma$, a
run of $\Au$ over $w$ is an infinite sequence of states
$q_0,q_1,\ldots$ such that $q_0\in Q_0$ and for all $i\geq 0$,
$(q_i,w(i),q_{i+1})\in \Delta$.
The run is accepting if for infinitely many $i$, $q_i\in\Acc$.  The
language $\Lang(\Au)$ accepted by $\Au$ consists of the infinite words
$w$ over $\Sigma$ such that there is an accepting run over $w$.

Fix a non-empty set $\Gamma$ of $\LTL$ formulas over $\AP$.
The closure $\cl(\Gamma)$ of $\Gamma$ is the set of $\LTL$ formulas
consisting of the sub-formulas of the formulas $\theta\in\Gamma$ and
their negations (we identify $\neg\neg\theta$ with $\theta$).
Note that $\Gamma\subseteq \cl(\Gamma)$.
Without loss of generality, we can assume that $\AP\subseteq \Gamma$.
Precisely, $\AP$ can be taken as the set of propositions occurring in
the given simple $\SHLTL$ sentence and $\cl(\Gamma)$ contains all the
propositions in $\AP$ and their negations.
For each formula $\theta\in \cl(\Gamma)\setminus \AP$, we introduce a
fresh atomic proposition not in $\AP$, denoted by
$at(\theta)$. Moreover, for allowing a uniform notation, for each
$p\in \AP$, we write $at(p)$ to mean $p$ itself.
Let $\AP_\Gamma$ be the set $\AP$ extended with these new
propositions.
By a straightforward adaptation of the well-known translation of
$\LTL$ formulas into equivalent $\NBA$~\cite{VardiW94}, we
obtain the following result, where for an infinite word $w$ over
$2^{\AP_\Gamma}$, $(w)_{\AP}$ denotes the projection of $w$ over $\AP$.

\newcounter{prop-NFAforLTL}
\setcounter{prop-NFAforLTL}{\value{proposition}}
\newcounter{sec-NFAforLTL}
\setcounter{sec-NFAforLTL}{\value{section}}

\begin{proposition}
  \label{prop:NFAforLTL}
  Given a finite set $\Gamma$ of $\LTL$ formulas over $\AP$, one can
  construct in single exponential time a $\NBA$ $\Au_\Gamma$
  over $2^{\AP_\Gamma}$ with $2^{O(|\AP_\Gamma|)}$ states satisfying
  the following:
  \begin{compactenum}
  \item let $w\in\Lang(\Au_\Gamma)$: then for all $i\geq 0$ and
    $\theta\in \cl(\Gamma)$, $at(\theta)\in w(i)$ if and only if
    $((w)_{\AP},i)\models \theta$.
  \item for each trace $\pi$ (i.e., infinite word over $2^{\AP}$),
    there exists $w\in \Lang(\Au_\Gamma)$ such that $\pi= (w)_{\AP}$.
\end{compactenum}
\end{proposition}

Let $\Ku=\tpl{S,S_0,E,V}$ be a finite Kripke structure over $\AP$ and $F\subseteq S$.
Next, we consider the synchronous product of the fair Kripke structure
$(\Ku , F)$ with the $\NBA$
$\Au_\Gamma = \tpl{2^{\AP_\Gamma},Q,Q_0,\Delta,\Acc}$ over
$2^{\AP_\Gamma}$ of Proposition~\ref{prop:NFAforLTL} associated with
$\Gamma$.
More specifically, we construct a Kripke structure $\Ku_\Gamma$ over
$\AP_\Gamma$ and a subset $F_\Gamma$ of $\Ku_\Gamma$-states such that
$\Lang(\Ku_\Gamma,F_\Gamma)$ is the set of words
$w\in\Lang(\Au_\Gamma)$ whose projections over $\AP$ are in
$\Lang(\Ku,F)$.
Formally, the \emph{$\Gamma$-extension of $(\Ku,F)$} is the fair
Kripke structure $(\Ku_\Gamma,F_\Gamma)$ where
$\Ku_\Gamma=\tpl{S_\Gamma,S_{0,\Gamma},E_\Gamma,V_\Gamma}$ and
$F_\Gamma$ are defined as follows where:
\begin{itemize}
\item $S_\Gamma$ is the set of tuples
  $(s,B,q,\ell)\in S\times 2^{\AP_\Gamma}\times Q\times \{1,2\}$ such
  that $V(s)= B\cap \AP$;
\item
  $S_{0,\Gamma}= S_\Gamma \cap (S_0\times 2^{\AP_\Gamma}\times
  Q_0\times \{1\})$;
\item $E_\Gamma$ consists of the following transitions:
  \begin{compactitem}
  \item $((s,B,q,1),(s',B',q',\ell))$ such that $(s,s')\in E$,
    $(q,B,q')\in \Delta$, and $\ell=2$ if $s\in F$ and $\ell=1$
    otherwise;
  \item $((s,B,q,2),(s',B',q',\ell))$ such that $(s,s')\in E$,
    $(q,B,q')\in \Delta$, and $\ell=1$ if $q\in \Acc$ and $\ell=2$
    otherwise.
  \end{compactitem}
\item for each $(s,B,q,\ell)\in S_\Gamma$, $V_\Gamma((s,B,q,\ell))=B$;
\item $F_\Gamma=\{(s,B,q,2)\in S_\Gamma\mid q\in \Acc\}$.
\end{itemize}
\noindent By construction and Proposition~\ref{prop:NFAforLTL}(2), we
easily obtain the following result.

\begin{proposition}
  \label{prop:SynchronousProduct}
  For each infinite word $w$ over $2^{\AP_\Gamma}$,
  $w\in \Lang(\Ku_\Gamma,F_\Gamma)$ if and only if
  $w\in\Lang(\Au_\Gamma)$ and $(w)_{\AP}\in \Lang(\Ku,F)$. Moreover,
  for each $\pi\in \Lang(\Ku,F)$, there exists
  $w\in\Lang(\Ku_\Gamma,F_\Gamma)$ such that $(w)_{\AP} =\pi$.
\end{proposition}

For each $\Gamma'\subseteq \Gamma$, let $\Gamma'_{prop}$ be the set of
propositions in $\AP_\Gamma$ associated with the formulas in
$\Gamma'$, in other words
$\Gamma'_{prop}\DefinedAs\{at(\theta)\mid \theta\in \Gamma'\}$.
By Propositions~\ref{prop:NFAforLTL}--\ref{prop:SynchronousProduct},
we deduce the following result which allows to reduce the fair model
checking against a simple $\SHLTL$ sentence to the fair model checking
against a $\SHLTL$ sentence in the fragment $\SHLTL[\Gamma'_{prop}]$
for some set $\Gamma'_{prop}$ of atomic propositions.

\begin{lemma}
  \label{lemma:EquivalenceKripkeSTructure}
  The following holds:
  \begin{compactenum}
  \item For each $\theta\in \Gamma$ and
    $w\in \Lang(\Ku_\Gamma,F_\Gamma)$, $at(\theta)\in w(0)$ iff
    $(w)_{\AP}\models \theta$.
  \item For all $\Gamma'\subseteq \Gamma$ and
    $w\in \Lang(\Ku_\Gamma,F_\Gamma)$,
    $(\stfr_{\Gamma'_{prop}}(w))_{\AP}= \stfr_{\Gamma'}(\pi)$ where
    $\pi=(w)_{\AP}$.
\end{compactenum}
\end{lemma}
\begin{proof}
  Property~1 directly follows from Proposition~\ref{prop:NFAforLTL}(1)
  and Proposition~\ref{prop:SynchronousProduct}.
  Now, let us consider Property~2.
  Let $\Gamma'\subseteq \Gamma$, $w\in \Lang(\Ku_\Gamma,F_\Gamma)$,
  and $\pi=(w)_{\AP}$.
  By Proposition~\ref{prop:SynchronousProduct},
  $w\in\Lang(\Au_\Gamma)$.
  Moreover, by Property~(1) of Proposition~\ref{prop:NFAforLTL}, for
  all $i\geq 0$ and $\theta\in\Gamma'$, $at(\theta)\in w(i)$ if and
  only if $(\pi,i)\models \theta$.
  Since $\Gamma'_{prop}\DefinedAs\{at(\theta)\mid \theta\in \Gamma'\}$, it
  follows that
  $(\stfr_{\Gamma'_{prop}}(w))_{\AP}=\stfr_{\Gamma'}(\pi)$, and the
  result follows.

\end{proof}

We can now prove the desired result.

\begin{theorem}
  \label{theorem:FormSimpleToProposition}
  Given a simple $\SHLTL$ sentence $\varphi$ and a fair finite Kripke
  structure $(\Ku,F)$ over $\AP$, one can construct in single
  exponential time in the size of $\varphi$, a $\SHLTL$ sentence
  $\varphi'$ having the same quantifier prefix as $\varphi$ and a fair
  finite Kripke structure $(\Ku',F')$ over an extension $\AP\,'$ of
  $\AP$ such that $|\varphi'|=O(|\varphi|)$, $\varphi'$ is in the
  fragment $\SHLTL[\AP\,'']$ for some $\AP\,''\subseteq \AP\,'$,
  $|\Ku'|=O(|\Ku|* 2^{O(|\varphi|)})$, and
  $\Lang(\Ku',F')\models \varphi'$ if and only if
  $\Lang(\Ku,F)\models \varphi$.
\end{theorem}

\begin{proof}
  By hypothesis, there is a finite set $\Gamma'$ of $\LTL$ formulas
  such that $\varphi$ is of the form
  \[
   Q_n x_n.Q_{n-1} x_{n-1}.\ldots.Q_1 x_1.\,\psi
 \]
 where $n\geq 1$, $Q_i\in\{\exists,\forall\}$ for all $i\in [1,n]$,
 and $\psi$ is a Boolean combination of quantifier-free formulas in a
 set $\Upsilon_1\cup \Upsilon_{\Gamma'}$, where $\Upsilon_1$ consists
 of \emph{one-variable} $\SHLTL$ quantifier-free formulas and
 $\Upsilon_{\Gamma'}$ consists of quantifier-free formulas in
 $\SHLTL[\Gamma']$.
 By Proposition~\ref{prop:FromSLTLtoLTL}, we can assume without loss
 of generality that the formulas in $\Upsilon_1$ are
 \emph{one-variable} $\HLTL$ quantifier-free
 formulas.
 Let $\LTL(\Upsilon_1)$ be the set of $\LTL$ formulas corresponding to
 the formulas in $\Upsilon_1$ (i.e., for each $\theta\in \Upsilon_1$,
 we take the $\LTL$ formula obtained from $\theta$ by removing the
 unique variable occurring in $\theta$).
 We define:
 \begin{compactitem}
 \item $\Gamma\DefinedAs \LTL(\Upsilon_1)\cup \Gamma'$. We assume that
   $\Gamma$ is not empty; otherwise the result is obvious.
 \item $(\Ku',F')\DefinedAs (\Ku_\Gamma,F_\Gamma)$, where
   $(\Ku_\Gamma,F_\Gamma)$ is the $\Gamma$-extension of $(\Ku,F)$;
 \item $\varphi'\DefinedAs Q_n x_n.\ldots.Q_1 x_1.\,\psi'$, where $\psi'$ is
   defined as follows:
   by hypothesis, $\psi$ can be seen as a propositional formula
   $\psi_p$ over the set of atomic formulas
   $\Upsilon_1\cup \Upsilon_{\Gamma'}$.
   Then, $\psi'$ is obtained from $\psi_p$ by replacing (i) each
   formula $\xi\in \Upsilon_1$ with the $x$-relativized proposition in
   $\AP_\Gamma$ given by $at(\LTL(\xi))[x]$, where $\LTL(\xi)$ is the
   $\LTL$ formula associated with $\xi$ and $x$ is the unique variable
   occurring in $\xi$, and (ii) each formula
   $\xi\in \Upsilon_{\Gamma'}$ with the formula obtained from $\xi$ by
   replacing each relativized temporal modality in $\xi$ with its
   $\Gamma'_{prop}$-relativized
   version. 
 \end{compactitem}\vspace{0.1cm}

 \noindent We show that $\Lang(\Ku_\Gamma,F_\Gamma)\models \varphi'$
 if and only if $\Lang(\Ku,F)\models \varphi$.
 Hence, Theorem~\ref{theorem:FormSimpleToProposition} directly
 follows.
 For each $i\in [1,n]$, let
 $\varphi_i\DefinedAs Q_i x_i.\ldots.Q_1 x_1.\,\psi$ and
 $\varphi'_i\DefinedAs Q_i x_i.\ldots.Q_1 x_1.\,\psi'$.
 Moreover, we write $\varphi_0$ (resp., $\varphi'_0$) to mean formula
 $\psi$ (resp., $\psi'$).
 The result directly follows from the following claim.\vspace{0.2cm}

 \noindent \textbf{Claim.}
 Let $0\leq i\leq n$ and
 $w_1,\ldots,w_{n-i}\in \Lang(\Ku_\Gamma,F_\Gamma)$.
 Then,
 $(\Lang(\Ku_\Gamma,F_\Gamma),\{x_1 \mapsto (w_1,0),\ldots,x_{n-i}
 \mapsto (w_{n-i},0) \})\models \varphi'_i$ if and only if
 $(\Lang(\Ku,F),\{x_1 \mapsto ((w_1)_{\AP},0),$
 $\ldots,x_{n-i} \mapsto ((w_{n-i})_{\AP},0) \})\models
 \varphi_i$.\vspace{0.1cm}

 \noindent \textbf{Proof of the Claim.} For the base case ($i=0$), the
 result directly follows from construction and
 Lemma~\ref{lemma:EquivalenceKripkeSTructure}.
 For the induction step, the result directly follows from the
 induction hypothesis and the second part of
 Proposition~\ref{prop:SynchronousProduct}.
\end{proof}

\noindent \textbf{Fair Model Checking against $\SHLTL[\Gamma]$ with
  $\Gamma\subseteq \AP$.} By
Theorem~\ref{theorem:FormSimpleToProposition}, we can restrict to
consider the fair model checking against the fragments
$\SHLTL[\Gamma]$ where $\Gamma$ is a non-empty finite set of atomic
propositions.
We show that this problem can be reduced in polynomial time to a
variant of model checking against $\HLTL$.

\begin{definition}[\LTL-conditioned model checking]
  For a Kripke structure $\Ku$ and a $\LTL$ formula $\theta$, we
  denote by $\Lang(\Ku,\theta)$ the set of traces of $\Ku$ which
  satisfy $\theta$.
  The \emph{$\LTL$-conditioned model checking problem against $\HLTL$}
  is checking for a finite Kripke structure $\Ku$, a $\LTL$ formula
  $\theta$ and a $\HLTL$ sentence $\varphi$, whether
  $\Lang(\Ku,\theta) \models \varphi$.
\end{definition}

\LTL-conditioned model checking against $\HLTL$ can be easily reduced
in linear time to $\HLTL$ model checking (for details
see the appendix).

\newcounter{prop-conditionedLTL}
\setcounter{prop-conditionedLTL}{\value{proposition}}
\newcounter{sec-conditionedLTL}
\setcounter{sec-conditionedLTL}{\value{section}}

\begin{proposition}
  \label{prop:conditionedLTL}
  Given an $\LTL$ formula $\theta$ and a $\HLTL$ sentence $\varphi$,
  one can construct in linear time a $\HLTL$ sentence $\varphi_\theta$
  having the same quantifier prefix as $\varphi$ such that for each
  Kripke structure $\Ku$, $\Lang(\Ku,\theta) \models \varphi$
  \emph{iff} $\Lang(\Ku)\models \varphi_\theta$.
 \end{proposition}

 Let $(\Ku,F)$ be a fair finite Kripke structure with
 $\Ku=\tpl{S,S_0,E,V}$ and $\varphi$ be a $\SHLTL[\Gamma]$ sentence
 with $\Gamma\subseteq \AP$ and $\Gamma\neq \emptyset$.
 Let $\acc$ be a fresh proposition not in $\AP$. Starting from $\Ku$,
 $F$, and $\Gamma$, we construct in polynomial time a finite Kripke
 structure $\widehat{\Ku}$ over $\widehat{\AP}=\AP\cup \{\acc\}$ and
 an $\LTL$ formula $\widehat{\theta}$ over $\widehat{\AP}$ such that
 the projections over $\AP$ of the traces of $\widehat{\Ku}$
 satisfying $\widehat{\theta}$ correspond to the traces in
 $\stfr_\Gamma(\Lang(\Ku,F))$.
 By Remark~\ref{remark:connection} and since $\varphi$ does not
 contain occurrences of the special proposition $\acc$, we obtain that
 $\Lang(\Ku,F)$ is a model of the $\SHLTL[\Gamma]$ sentence $\varphi$
 \emph{iff} $\Lang(\widehat{K},\widehat{\theta})$ is a model of the
 $\HLTL$ sentence $\HLTL(\varphi)$.

 Intuitively, the Kripke structure $\widehat{\Ku}$ is obtained from
 $\Ku$ by adding edges which keep track of the states associated with
 the starting positions of adjacent segments along the
 $\Gamma$-stutter factorizations of (the traces of) the $F$-fair paths
 of $\Ku$. Formally, let $R_\Gamma(\Ku)$ and $R_\Gamma(\Ku,F)$ be the
 sets of state pairs in $\Ku$ defined as follows:
 \begin{itemize}
 \item $R_\Gamma(\Ku)$ consists of the pairs $(q,q')\in S\times S$
   such that $V(q)\cap \Gamma\neq V(q')\cap \Gamma$ and there is a
   finite path of $\Ku$ of the form $q\cdot \rho \cdot q'$ such that
   $V(q)\cap \Gamma= V(\rho(i))\cap \Gamma$ for all $0\leq i<|\rho|$.
 \item $R_\Gamma(\Ku,F)$ is defined similarly but, additionally, we
   require that the finite path $q\cdot \rho \cdot q'$ visits some
   accepting state in $F$.
 \end{itemize}
 The finite sets $R_\Gamma(\Ku)$ and $R_\Gamma(\Ku,F)$ can be easily
 computed in polynomial time by standard closure algorithms. By
 exploiting the sets $R_\Gamma(\Ku)$ and $R_\Gamma(\Ku,F)$, we define
 the finite Kripke structure
 $\widehat{\Ku}=\tpl{\widehat{S},\widehat{S_0},\widehat{E},\widehat{V}}$
 over $\widehat{\AP}=\AP\cup \{\acc\}$ as follows:
\begin{itemize}
\item $\widehat{S}=S\times \{0,1\}$ and
  $\widehat{S_0}=S_0\times\{0\}$.
\item $\widehat{E}$ consists of the edges $((s,\ell),(s,\ell'))$ such
  that one of the following holds:
  \begin{compactitem}
  \item \emph{either} $(s,s')\in E \cup R_\Gamma(\Ku)$ and ($\ell'=1$
    iff $s'\in F$),
  \item \emph{or} $(s,s')\in R_\Gamma(\Ku,F)$ and $\ell'=1$.
   \end{compactitem}
 \item $\widehat{V}((s,1))=V(s)\cup\{\acc\}$ and
   $\widehat{V}((s,0))=V(s)$. 
 \end{itemize}

 \noindent Let $\widehat{\theta}$ be the $\LTL$ formula over
 $\widehat{\AP}$ defined as follows:
\[
\Always\Eventually\acc \wedge \Always\bigl(\displaystyle{\bigvee_{p\in \Gamma}} (p \leftrightarrow \neg \Next p)\, \vee\,
\displaystyle{\bigwedge_{p\in \Gamma}} \Always(p \leftrightarrow \Next p)   \bigr)
\]
The first conjunct in the definition of $\widehat{\theta}$ ensures
that proposition $\acc$ holds infinitely often while the second
conjunct captures the traces that are $\Gamma$-strutter free.
By construction, we easily obtain the following result.

\begin{proposition}
  \label{prop:DecidabilityPropositional}
  $\stfr_\Gamma(\Lang(\Ku,F))$ coincides with the set of projections
  over $\AP$ of the traces in $\Lang(\widehat{\Ku},\widehat{\theta})$.
\end{proposition}
\begin{proof}
  Let $\pi\in\stfr_\Gamma(\Lang(\Ku,F))$.
  Hence, there is a $F$-fair path $\nu$ of $\Ku$ such that
  $\pi= \stfr_\Gamma(V(\nu))$ where $V(\nu)$ is the trace associated
  with $\nu$.
  By construction, there is a trace $\widehat{\pi}$ of $\widehat{\Ku}$
  such that $\acc\in \widehat{\pi}(i)$ for infinitely many $i$ and the
  projection of $\widehat{\pi}$ over $\AP$ coincides with $\pi$.
  By construction of the $\LTL$ formula $\widehat{\theta}$,
  $\widehat{\pi}$ satisfies $\widehat{\theta}$.
  Hence, $\widehat{\pi}\in \Lang(\widehat{\Ku},\widehat{\theta})$.
  The converse direction is similar.
 \end{proof}

 By Proposition~\ref{prop:DecidabilityPropositional} and
 Remark~\ref{remark:connection}, we obtain that for each
 $\SHLTL[\Gamma]$ sentence $\varphi$, $\Lang(\Ku,F)\models \varphi$
 iff $\Lang(\widehat{K},\widehat{\theta})\models \HLTL(\varphi)$.

 By~\cite{Rabe2016} model checking a finite Kripke structure $\Ku$
 against a $\HLTL$ sentence $\varphi$ of quantifier alternation depth
 $d$ can be done in nondeterministic space bounded by
 $O(\Tower_2(d,|\varphi|\log(|\Ku|)))$.
 Thus, since simple $\SHLTL$ subsumes $\HLTL$, by
 Theorem~\ref{theorem:FormSimpleToProposition} and
 Proposition~\ref{prop:conditionedLTL}, we obtain the main result of
 this section, where the lower bounds correspond to the known ones for
 $\HLTL$~\cite{Rabe2016}.

 \begin{theorem}
   For each $d\in \N$, (fair) model checking against simple $\SHLTL$
   sentences of quantifier alternation depth $d$ is
   $d$-\EXPSPACE-complete, and for a fixed formula, it is
   $(d-1)$-\EXPSPACE-complete for $d>0$ and \NLOGSPACE-complete
   otherwise.
\end{theorem}


\section{Context $\HLTL$}\label{sec:ContextHyper}

In this section, we introduce an alternative logical framework for
specifying asynchronous linear-time hyperproperties.
The novel framework, we call \emph{context} $\HLTL$ ($\CHLTL$ for
short), extends $\HLTL$ by unary modalities $\tpl{C}$ parameterized by
a non-empty subset $C$ of trace variables---called the
\emph{context}---which restrict the evaluation of the temporal
modalities to the traces associated with the variables in $C$.
Formally, $\CHLTL$ formulas over the given finite set $\AP$ of atomic
propositions and finite set $\Var$ of trace variables are linear-time
hyper expressions over multi-trace specifications $\psi$, called
\emph{$\CHLTL$ quantifier-free formulas}, where $\psi$ is defined by
the following syntax:
\[
   \psi ::=  \top  \ | \   \Rel{p}{x}  \ | \ \neg \psi \ | \ \psi \wedge \psi \ | \ \Next\psi \ | \   \psi \Until \psi \ | \ \tpl{C} \psi
\]
where $p\in \AP$, $x\in \Var$, and $\tpl{C}$ is the context modality
with $\emptyset \neq C\subseteq \Var$.
A \emph{context} is a non-empty subset of trace variables in
$\Var$.
The size $|\xi|$ of a $\CHLTL$ (quantifier-free) formula $\xi$ is the
number of distinct sub-formulas of $\xi$.
A context $C$ is \emph{global for a formula $\xi$} if $C$ contains all
the trace variables occurring in $\xi$.  \vspace{0.2cm}

\noindent \textbf{Semantics of $\CHLTL$ quantifier-free formulas.} Let
$\Pi$ be a pointed trace assignment.
Given a context $C$ and an offset $i\geq 0$, we denote by $\Pi +_C i$
the pointed trace assignment with domain $\Dom(\Pi)$ defined as
follows:
\begin{compactitem}
\item for each $x\in \Dom(\Pi)\cap C$ with $\Pi(x)=(\pi,h)$,
  $[\Pi +_C i](x)=(\pi,h+i)$;
\item for each $x\in \Dom(\Pi)\setminus C$, $[\Pi +_C i](x)=\Pi(x)$.
\end{compactitem}
Intuitively, the positions of the pointed traces associated with the
variables in $C$ advance of the offset $i$, while the positions of the
other pointed traces remain unchanged.

Given a $\CHLTL$ quantifier-free formula $\psi$, a context $C$, and a
pointed trace assignment $\Pi$ such that $\Dom(\Pi)$ contains the
trace variables occurring in $\psi$, the satisfaction relation
$(\Pi,C)\models \psi$ is inductively defined as follows (we omit the
semantics of the Boolean connectives which is standard):
  \[ \begin{array}{ll}
 (\Pi,C) \models \Rel{p}{x} & \Leftrightarrow  \Pi(x)=(\pi,i) \text{ and }p\in \pi(i)\\
  (\Pi,C) \models  \Next\psi & \Leftrightarrow (\Pi +_C 1,C)\models  \psi\\
  (\Pi,C) \models  \psi_1\Until \psi_2 & \Leftrightarrow  \text{for some }i\geq 0:\,  (\Pi +_C i,C)\models  \psi_2 \\
     &\phantom{\Leftrightarrow} \text{and }(\Pi +_C k,C) \models  \psi_1 \text{ for all } k<i\\
 (\Pi,C) \models \tpl{C'}\psi&\Leftrightarrow (\Pi,C') \models  \psi
\end{array} \]
We write $\Pi\models \psi$ to mean that $(\Pi,\Var)\models\psi$.
\vspace{0.2cm}

\noindent \textbf{Examples of specifications.}
Context $\HLTL$ extends $\HLTL$ by allowing to specify complex
combinations of asynchronous and synchronous constraints.
As an example consider the following property~\cite{GutsfeldOO21}.
A a $\HLTL$ quantifier-free formula $\psi(x_1,\ldots,x_n)$ holds along
the traces bound by variables $x_1\ldots,x_n$ after an initialization
phase, which can take a different number of steps on each trace.
This can be expressed by an $\CHLTL$ quantifier-free formula as
follows, where proposition $in$ characterizes the initialization
phase:
  \[ \begin{array}{l}
\tpl{\{x_1\}}(in[x_1] \Until (\neg in[x_1] \wedge \tpl{\{x_2\}}(\ldots \vspace{0.1cm}\\
\hspace{1cm}\tpl{\{x_n\}}(in[x_n] \Until (\neg in[x_n]\wedge \tpl{\{x_1,\ldots,x_n\}}\psi  )) \ldots )))
\end{array} \]
As another example, illustrating the high expressiveness of $\CHLTL$,
we consider the following hyper-bounded-time response requirement:
``\emph{for every trace there is a bound $k$ such that each request
  $q$ is followed by a response $p$ within exactly $k$ steps.}''
This can be expressed in $\CHLTL$ as follows:
%
%
 \[ \begin{array}{l}
     \forall x.\,\exists y.\, \Eventually \Rel{q}{x} \longrightarrow  \displaystyle{\bigl[\bigwedge_{r\in\AP}\Always(r[x] \leftrightarrow r[y])}\,\wedge \vspace{0.1cm} \\
      \hspace{0.2cm}\tpl{\{y\}}\Eventually
      \begin{pmatrix}
        \begin{array}{l}
          \Rel{q}{y} \wedge     \tpl{\{x\}} \Always
          \begin{pmatrix}
            \begin{array}{l}
              \Rel{q}{x} \rightarrow \\
              \{x,y\} (\neg \Rel{p}{y}\Until \Rel{p}{x})              
            \end{array}
          \end{pmatrix}
        \end{array}
      \end{pmatrix}\bigr]
\end{array} \]

Note that $x$ and $y$ refer to the same trace and the context
modalities are exploited to synchronously compare distinct segments
along the same trace, that correspond to different request-response
intervals.
This ability is not supported by Stuttering $\HLTL$.
On the other hand, we conjecture that unlike $\SHLTL$, $\CHLTL$ cannot
express asynchronous variants of security properties such as
noninterference and observational determinism (see
Subsection~\ref{sec:SyntaxSemanticsHLTLS}).

\subsection{Undecidability of model checking against $\CHLTL$}\label{sec:UndeciablilityContextHLTL}

In this section, we establish that model checking against $\CHLTL$ is
in general undecidable.
Let $\F_0$ and $\F_1$ be the fragments of $\CHLTL$ consisting of the
formulas such that the number of trace variables is $2$, the nesting
depth of context modalities is $2$, and, additionally,
\begin{inparaenum}[(i)]
\item in $\F_0$ the quantifier alternation depth is $0$, and
\item in $\F_1$ the quantifier alternation depth is $1$ and each
  temporal modality in the scope of a non-global context is
  $\Eventually$.
\end{inparaenum}

\begin{theorem}\label{theo:UndecidabilityContextHyeprLTL}
  The model checking problem against $\CHLTL$ is undecidable even for
  the fragments $\F_0$ and $\F_1$.
\end{theorem}

Theorem~\ref{theo:UndecidabilityContextHyeprLTL} is proved by a
polynomial-time reduction from the \emph{halting problem for Minsky
  $2$-counter machines}~\cite{Minsky67}.
Such a machine is a tuple $M = \tpl{Q,q_\init,q_\halt,$ $\Delta}$,
where $Q$ is a finite set of (control) locations, $q_\init\in Q$ is
the initial location, $q_\halt\in Q$ is the halting location, and
$\Delta \subseteq Q\times \Inst \times Q$ is a transition relation
over the instruction set $\Inst= \{\inc,\dec,\zero\}\times \{1,2\}$.
We adopt the following notational conventions.
For an instruction $\instr=(\_\,,c)\in \Inst$, let
$c(\instr)\DefinedAs c\in\{1,2\}$ be the \emph{counter} associated with
$\instr$.
For a transition $\delta\in \Delta$ of the form
$\delta=(q,\instr,q')$, we define $\From(\delta)\DefinedAs q$,
$\instr(\delta)\DefinedAs\instr$, $c(\delta)\DefinedAs c(\instr)$, and
$\To(\delta)\DefinedAs q'$.
Without loss of generality, we assume that for each transition
$\delta\in \Delta$, $\From(\delta)\neq q_\halt$.

An $M$-configuration is a pair $(q,\nu)$ consisting of a location
$q\in Q$ and a counter valuation $\nu: \{1,2\}\to \N$.
A computation of $M$ is a non-empty \emph{finite} sequence
$(q_1,\nu_1),\ldots ,(q_k,\nu_k)$ of configurations such that for each
$1\leq i<k$, $(q_i,\instr,q_{i+1})\in \Delta$ for some instruction
$\instr\in \Inst$ (depending on $i$) and the following holds, where
$c\in \{1,2\}$ is the counter associated with the instruction
$\instr$:
\begin{inparaenum}[(i)]
\item $\nu_{i+1}(c')= \nu_i(c')$ if $c'\neq c$;
\item $\nu_{i+1}(c)= \nu_i(c) +1$ if $\instr=(\inc,c)$;
\item $\nu_{i+1}(c)= \nu_i(c) -1$ if $\instr=(\dec,c)$ (in particular,
  it has to be $\nu_i(c)>0$); and
\item $\nu_{i+1}(c)= \nu_i(c)=0$ if $\instr=(\zero,c)$.
\end{inparaenum}
$M$ \emph{halts} if there is a computation starting at the
\emph{initial} configuration $(q_\init, \nu_\init)$, where
$\nu_\init(1) = \nu_\init(2) = 0$, and leading to some halting
configuration $(q_{\halt}, \nu)$.
The halting problem is to decide whether a given machine $M$ halts,
and it is undecidable~\cite{Minsky67}.
We prove the following result, from which
Theorem~\ref{theo:UndecidabilityContextHyeprLTL} for the fragment
$\F_0$ directly follows.

\begin{proposition}\label{prop:undecidabilityContextHLTL}
  One can build 
  a finite Kripke Structure $\Ku_M$ and a $\CHLTL$ sentence
  $\varphi_M$ in the fragment $\F_0$ such that $M$ halts \emph{iff}
  $\Ku_M\models \varphi_M$.
\end{proposition}

\begin{proof}
  First, we define a suitable encoding of a computation of $M$ as a
  trace where the finite set $\AP$ of atomic propositions is given by
  $ \AP\DefinedAs \Delta \cup \{1,2,\Beg_1,\Beg_2\} $.

  Intuitively, in the encoding of an $M$-computation, we keep track of
  the transition used in the current step of the
  computation.
  Moreover, for each $c\in \{1,2\}$, the propositions in
  $\{c,\Beg_c\}$ are used for encoding the current value of counter
  $c$.
  In particular, for $c\in \{1,2\}$, a \emph{$c$-code for the
    $M$-transition $\delta\in\Delta$} is a finite word $w_c$ over
  $2^{\AP}$ of the form $\{\delta,\Beg_c\}\cdot \{\delta,c\}^{h}$ for
  some $h\geq 0$ such that $h=0$ if $\instr(\delta)= (\zero,c)$.
  The $c$-code $w_c$ encodes the value for counter $c$ given by $h$
  (or equivalently $|w_c|-1$).
  {Note that only the occurrences of the symbols $\{\delta,c\}$ encode
    units in the value of counter $c$, while the symbol
    $\{\delta,\Beg_c\}$ is only used as left marker in the encoding.}
  A \emph{configuration-code $w$ for the $M$-transition
    $\delta\in\Delta$} is a finite word over $2^{\AP}$ of the form
  $w= \{\delta\} \cdot w_1 \cdot w_2 $ such that for each counter
  $c\in \{1,2\}$, $w_c$ is a $c$-code for transition $\delta$.
  The configuration-code $w$ encodes the $M$-configuration
  $(\From(\delta),\nu)$, where $\nu(c)=|w_c|-1$ for all
  $c\in \{1,2\}$. Note that if $\instr(\delta)=(\zero,c)$, then
  $\nu(c)=0$.

  A \emph{computation}-code is a trace of the form
  $\pi= w_{ \delta_1} \cdots w_{ \delta_k} \cdot \emptyset^{\omega}$,
  where $k\geq 1$ and for all $1\leq i\leq k$, $w_{ \delta_i}$ is a
  configuration-code for transition $ \delta_i$, and whenever $i<k$,
  it holds that $\To(\delta_i)=\From(\delta_{i+1})$.
  Note that by our assumptions $\To(\delta_i)\neq q_\halt$ for all
  $1\leq i<k$.
  The computation-code $\pi$ is \emph{initial} if the first
  configuration-code $w_{\delta_1}$ encodes the initial configuration,
  and it is \emph{halting} if for the last configuration-code
  $w_{\delta_k}$ in $\pi$, it holds that $\To(\delta_k)=q_\halt$.
  For all $1\leq i\leq k$, let $(q_i,\nu_i)$ be the $M$-configuration
  encoded by the configuration-code $w_{\delta_i}$ and
  $c_i= c(\delta_i)$.
  The computation-code $\pi$ is \emph{good} if, additionally, for all
  $1\leq j< k$, the following holds:
  \begin{inparaenum}[(i)]
  \item $\nu_{j+1}(c)=\nu_j(c)$ if either $c \neq c_j$ or $\instr(\delta_j)= (\zero,c_j)$  (\emph{equality requirement});
  \item $\nu_{j+1}(c_j)= \nu_j(c_j)+1$ if $\instr(\delta_j)= (\inc,c_j)$ (\emph{increment requirement});
  \item $\nu_{j+1}(c_j)= \nu_j(c_j)-1$ if $\instr(\delta_j)= (\dec,c_j)$ (\emph{decrement requirement}).
\end{inparaenum}

\noindent Clearly, $M$ halts \emph{iff} there exists an initial and
halting good computation-code.
By construction, it is a trivial task to define a Kripke structure
$\Ku_M$ satisfying the following.\vspace{0.2cm}

\noindent \textbf{Claim.} One can construct in polynomial time a
finite Kripke structure $\Ku_M$ over $\AP$ such that the set of traces
of $\Ku_M$ which visit some empty position (\ie, a position with
label the empty set of propositions) corresponds to the set of initial
and halting computation-codes.\vspace{0.2cm}

We now define a $\CHLTL$ sentence $\varphi_M$ in the fragment $\F_0$
that, when interpreted on the Kripke structure $\Ku_M$, captures the
traces $\pi$ of $\Ku_M$ which visit some empty position (hence, by the
previous claim, $\pi$ is an initial and halting computation-code) and
satisfy the goodness requirement.
\[
  \varphi_M\DefinedAs \exists x_1.\,\exists x_2.\, \Always
  \displaystyle{\bigwedge_{p\in\AP}}(\Rel{p}{x_1}\leftrightarrow
  \Rel{p}{x_2})\wedge \Eventually\displaystyle{\bigwedge_{p\in\AP}}\neg
  \Rel{p}{x_1}\wedge \psi_{good}
\]
where the $\CHLTL$ quantifier-free sub-formula $\psi_{good}$ is
defined in the following.
Intuitively, when interpreted on the Kripke structure $\Ku_M$ of the
previous claim, formula $\varphi_M$ asserts the existence of two
traces $\pi_1$ and $\pi_2$ bounded to the trace variables $x_1$ and
$x_2$, respectively, such that
\begin{inparaenum}[(i)]
\item $\pi_1$ and $\pi_2$ coincide (this is ensured by the first
  conjunct);
\item $\pi_1$ is an initial and halting computation-code (this is
  ensured by the previous claim and the second conjunct);
\item $\pi_1$ satisfies the goodness requirement by means of the
  conjunct $\psi_{good}$.
 \end{inparaenum}

 We now define the quantifier-free formula $\psi_{good}$.
 Let $\Delta_\halt\DefinedAs\{\delta\in\Delta\mid \To(\delta)=q_\halt\}$ be
 the set of transitions having as a target location the halting
 location.
 In the definition of $\psi_{good}$, we crucially exploit the context
 modalities.
 Essentially, for each position $i\geq 0$ along $\pi_1$ and $\pi_2$
 corresponding to the initial position of a $c$-code for a transition
 $\delta\notin \Delta_\halt$ within a configuration code $w_\delta$,
 we exploit:
 \begin{compactitem}
 \item temporal modalities in the scope of the context modality
   $\tpl{\{x_2\}}$ for moving the current position along trace $\pi_2$
   (the trace bounded by $x_2$) to the beginning of the $c$-code of
   the configuration code $w'$ following $w_\delta$,
 \item and then we use the temporal modalities in the scope of the
   global context $\{x_1,x_2\}$ for synchronously ensuring that for
   the $c$-codes associated to the consecutive configuration codes
   $w_{\delta}$ and $w'$, the equality, increment, and decrement
   requirements are fulfilled.
 \end{compactitem}\vspace{0.1cm}

 \noindent Formally, the $\CHLTL$ quantifier-free formula
 $\psi_{good}$ is defined as follows:
\[
 \begin{array}{l}
\psi_{good}\DefinedAs   \Always\displaystyle{\bigwedge_{\delta\in\Delta\setminus \Delta_\halt}\bigwedge_{c\in\{1,2\}}}\Bigl[(\delta[x_1]\wedge \Beg_c[x_1])\longrightarrow \\
\hspace{1cm}\tpl{\{x_2\}}\Next\Bigl(\neg\Beg_c[x_2] \Until \bigl(\Beg_c[x_2] \wedge \\
\hspace{1cm} \tpl{\{x_1,x_2\}}(\psi_{=}(\delta,c)\wedge \psi_{inc}(\delta,c)\wedge \psi_{dec}(\delta,c))\bigr) \Bigr)\Bigr]
\end{array}
\]
where the sub-formulas $\psi_{=}(\delta,c)$, $\psi_{inc}(\delta,c)$,
and $\psi_{dec}(\delta,c)$ capture the equality, increment, and
decrement requirement, respectively, and are defined as follows.
\[
\begin{array}{r@{\;}c@{\;}l}
\psi_{=}(\delta,c)&\DefinedAs & [c\neq c(\delta)\vee  \instr(\delta)= (\zero,c)]\longrightarrow  \\[0.5em]
&& \hspace{-1cm} \Next [(c[x_1] \wedge c[x_2])\Until (\neg c[x_1]\wedge \neg c[x_2])]\vspace{0.1cm}\\
\psi_{inc}(\delta,c)&\DefinedAs & \instr(\delta)= (\inc,c)\longrightarrow  \\[0.5em]
&& \hspace{-1cm}\Next [(c[x_1] \wedge c[x_2])\Until (\neg c[x_1]\wedge c[x_2] \wedge \Next \neg c[x_2])]\vspace{0.1cm}\\
\psi_{dec}(\delta,c)&\DefinedAs & \instr(\delta)= (\dec,c)\longrightarrow \\[0.5em]
&& \hspace{-1cm} \Next [(c[x_1] \wedge c[x_2])\Until (c[x_1]\wedge \neg c[x_2] \wedge \Next \neg c[x_1])]
\end{array}
\]
This finishes the proof.
\end{proof}

\subsection{Fragment of $\CHLTL$ with decidable model checking}

By Theorem~\ref{theo:UndecidabilityContextHyeprLTL}, model checking 
$\CHLTL$ is undecidable even for formulas where $\Eventually$ is the
unique temporal modality occurring in the scope of a non-global
context operator.
This justifies the investigation of the fragment, we call
\emph{bounded $\CHLTL$}, consisting of the $\CHLTL$ formulas where the
unique temporal modality occurring in a non-global context is the next
modality $\Next$.
For instance, for each $k\geq 0$, the formula
$\tpl{\{x_1\}}\Next^{k}(\tpl{\{x_1,x_2\}}\Always(\Rel{p}{x_1}\leftrightarrow
\Rel{p}{x_2}))$ is bounded while the formula
$\tpl{\{x_1\}}\Eventually(\tpl{\{x_1,x_2\}}\Always(\Rel{p}{x_1}\leftrightarrow
\Rel{p}{x_2}))$ is not.
Note that bounded $\CHLTL$ subsumes $\HLTL$ and is able to express a
restricted form of asynchronicity by allowing to compare traces at
different timestamps whose distances are bounded (a bound is given by
the nesting depth of next modalities in the formula).
As an example, the after-initialization synchronization requirement
described after the definition of $\CHLTL$ can be expressed by
assuming that the lengths of the initialization phases differ at most
a given integer $k$.
We conjecture that bounded $\CHLTL$ is not more expressive than
$\HLTL$.
However, as a consequence of Theorem~\ref{theo:ComplexityBoundedCHLTL}
below, for a fixed quantifier alternation depth, bounded $\CHLTL$ is
at least singly exponentially more succinct than $\HLTL$.

We show that model checking against bounded $\CHLTL$ is decidable by a
polynomial-time translation of bounded $\CHLTL$ quantifier-free
formulas $\psi$ into equivalent $(|\psi|+1)$-synchronous B\"{u}chi
$\AAWA$.

\newcounter{prop-FromContextHyperToAAWA}
\setcounter{prop-FromContextHyperToAAWA}{\value{proposition}}
\newcounter{sec-FromContextHyperToAAWA}
\setcounter{sec-FromContextHyperToAAWA}{\value{section}}

\begin{proposition}
  \label{prop:FromContextHyperToAAWA}
  Given a $\CHLTL$ quantifier-free formula $\psi$ with trace variables
  $x_1,\ldots,x_n$, one can build in polynomial time a B\"{u}chi
  $\NAAWA$ $\Au_\psi$ such that $\Lang(\Au_\psi)$ is the set of
  $n$-tuples $(\pi_1,\ldots,\pi_n)$ of traces so that
  $(\{x_1\mapsto (\pi_1,0),\ldots,x_1\mapsto
  (\pi_n,0)\},\{x_1,\ldots,x_n\})\models \psi$. Moreover, $\Au_\psi$
  is $(|\psi|+1)$-synchronous if $\psi$ is in the bounded fragment of
  $\CHLTL$.
\end{proposition}

\begin{proof} By exploiting the release modality $\Release$ (the dual
  of the until modality), we can assume without loss of generality
  that $\psi$ is in negation normal form, so negation is applied only
  to relativized atomic propositions.
  The construction of the B\"{u}chi $\NAAWA$ $\Au_\psi$ is a
  generalization of the standard translation of $\LTL$ formulas into
  equivalent standard B\"{u}chi alternating word automata.
  In particular, the automaton $\Au_\psi$ keeps track in its state of
  the sub-formula of $\psi$ currently processed, of the current
  context $C$, and of a counter modulo the cardinality $|C|$ of
  $C$.
  This counter is used for recording the directions associated to the
  variables in $C$ for which a move of one position to the right has
  already been done in the current phase of $|C|$-steps.
  \details{The transition function reflects the `local'
    characterization of the semantics of the Boolean connectives and
    the temporal and context modalities, while the B\"{u}chi
    acceptance condition is used for ensuring the fulfillment of the
    liveness requirements $\theta_2$ in the until sub-formulas
    $\theta_1\Until \theta_2$ of $\psi$.
  }
  By construction, whenever the automaton is in a state associated
  with a sub-formula $\theta$ of $\psi$, then $\Au_\psi$ can move only
  to states associated with $\theta$ or with strict sub-formulas of
  $\theta$.
  In particular, each path in a run of $\Au_\psi$ can be factorized
  into a finite number $\nu_1,\ldots,\nu_k$ of contiguous segments
  (with $\nu_k$ possibly infinite) such that for each $i\in [1,k]$,
  segment $\nu_i$ is associated with a sub-formula $\theta_i$ of
  $\psi$ and a context $C_i$ occurring in $\psi$, and the following
  holds, where the \emph{offset} of a position vector
  $\wp=(j_1,\ldots,j_n)$ in $\N^{n}$ is the maximum over the
  differences between pairs of components, i.e.
  $\max(\{j_\ell-j_{\ell'}\mid \ell,\ell'\in [1,n]\})$:
\begin{itemize}
\item there is some occurrence of $\theta_i$ in $\psi$ which is in the
  scope of the context modality $\tpl{C_i}$;
\item if $i<k$, then $\theta_{i+1}$ is a strict sub-formula of $\theta_i$;
\item if either $C_i$ is global or the root modality of $\theta_i$ is
  not in $\{\Until,\Release\}$, then the offset at each node along the
  segment $\nu_i$ and at the first node of $\nu_{i+1}$ if $i<k$ is at
  most the offset at the beginning of $\nu_i$ plus one.
\end{itemize}
\noindent Hence, if $\psi$ is in the bounded fragment of $\CHLTL$, the
offset at each node of a run is at most $|\psi|+1$, i.e. $\Au_\psi$ is
$(|\psi|+1)$-synchronous and the result follows.
\end{proof}

By exploiting Propositions~\ref{prop:KsynchronousAWA}
and~\ref{prop:FromContextHyperToAAWA}, we deduce that
for a fixed quantifier alternation depth $d$, model checking against
bounded $\CHLTL$ is $(d+1)$-\EXPSPACE-complete, hence singly
exponentially harder than model checking against $\HLTL$.
However, for a fixed formula, the complexity of the problem is the
same as for $\HLTL$.

\begin{theorem}\label{theo:ComplexityBoundedCHLTL}
  Let $d\in \N$.
  The (fair) model checking problem against bounded $\CHLTL$ sentences
  of quantifier alternation depth $d$ is $(d+1)$-\EXPSPACE-complete,
  and for a fixed formula, it is $(d-1)$-\EXPSPACE-complete for $d>0$
  and \NLOGSPACE-complete otherwise.
  
\end{theorem}

\begin{proof}
  The upper bounds follow from Propositions~\ref{prop:KsynchronousAWA}
  and~\ref{prop:FromContextHyperToAAWA}, while since bounded $\CHLTL$
  subsumes $\HLTL$, the lower bound for a fixed formula of alternation
  depth $d$ is inherited from the known one for
  $\HLTL$~\cite{Rabe2016}.
  Finally, for $(d+1)$-\EXPSPACE-hardness, we adapt the reduction
  given in~\cite{Rabe2016} for showing that for all integer constants
  $c>1$ and $c'\geq 1$, model checking against $\HLTL$ sentences
  $\varphi$ with quantifier alternation depth $d$ requires space at
  least $\Omega(\Tower_c(d,|\varphi|^{c'}))$.
  Here, for simplicity, we assume that $c=2$ and $c'=1$.
  The reduction in~\cite{Rabe2016} for model checking $\HLTL$ is based
  on building, for each $n>1$, an $\HLTL$ formula of size polynomial
  in $n$, with quantifier alternation depth $d$ over a singleton set
  $\AP=\{p\}$ of atomic propositions.
  This formula is of the form $\psi_d(x,y)$ for two free trace
  variables $x$ and $y$ such that for all traces $\pi_x$ and $\pi_y$
  (over $\AP$),
  $\{x \mapsto (\pi_x,0),y \mapsto (\pi_y,0)\}\models \psi_d(x,y)$
  \emph{if and only if} $p$ occurs exactly once in $\pi_x$ (resp.,
  $\pi_y$) and $p$ occurs on $\pi_y$ exactly $g(d+1,n)$ positions
  after $p$ occurs on $\pi_x$, where
  \begin{compactitem}
  \item  $g(0,n)= \Tower_2(0,n)=n$;
  \item $g(d+1,n)= g(d,n)*\Tower_2(d+1,n)$.
  \end{compactitem}\vspace{0.1cm}

  \noindent The construction is given by induction on $d$, and the
  formula $\psi_0(x,y)$ for the base case $d=0$ and a fixed $n>1$ do
  not use universal quantifiers (note that $\psi_0(x,y)$ requires that
  $\Rel{p}{y}$ occurs exactly $n * 2^{n}$ positions after $p[x]$
  occurs).
  Thus, since bounded $\CHLTL$ subsumes $\HLTL$ and
  $g(2,n)= n * 2^{n}* 2^{2^{n}}$, in order to show that model checking
  against bounded $\CHLTL$ formulas $\varphi$ with quantifier
  alternation depth $d$ requires space at least
  $\Omega(\Tower_2(d+1,|\varphi|))$, it suffices to show the following
  result.

  \noindent \textbf{Claim.} Let $\AP=\{p\}$ and $n>1$.
  One construct in time polynomial in $n$ a bounded $\CHLTL$ formula
  $\psi(x,y)$ with two free variables $x$ and $y$ and not containing
  universal quantifiers (hence, the quantifier alternation depth is
  $0$) such that for all traces $\pi_x$ and $\pi_y$,
  $\{x \mapsto (\pi_x,0),y \mapsto (\pi_y,0)\}\models \psi(x,y)$ iff
\begin{compactitem}
\item $p$ occurs exactly once on $\pi_x$ (resp., $\pi_y$);
\item for each $i\geq$, $p\in \pi_x(i)$ iff
  $p\in \pi_y(i+n*2^{n}* 2^{2^{n}})$.
\end{compactitem}
\end{proof}


\section{Conclusions}

We have introduced in this paper two extensions of HyperLTL to express
asynchronous hyperproperties: $\SHLTL$ and $\CHLTL$.
Even though the model-checking problems of these logics are in general
undecidable we have presented one decidable fragment of each logic
that allows to express asynchronous properties of interest.

We plan to extend our work in many directions.
First, we intend to settle the question concerning the comparison of
the expressive power of $\SHLTL$ and $\CHLTL$.
Second, we aim to understand the decidability border of model checking
syntactical fragments of the framework resulting by combining $\SHLTL$
and $\CHLTL$.
In particular, the decidability status of model checking against the
fragment obtained by merging simple $\SHLTL$ and bounded $\SHLTL$ is
open.
Finally, other goals regard the extensions of the considered logic to
the branching-time setting and the investigation of first-order and
monadic second-order logics for the specification of asynchronous
hyperproperties in the linear-time and branching-time settings.

\bibliographystyle{IEEEtran}
\bibliography{IEEEabrv,biblio}

\begin{thebibliography}{10}
\providecommand{\url}[1]{#1}
\csname url@samestyle\endcsname
\providecommand{\newblock}{\relax}
\providecommand{\bibinfo}[2]{#2}
\providecommand{\BIBentrySTDinterwordspacing}{\spaceskip=0pt\relax}
\providecommand{\BIBentryALTinterwordstretchfactor}{4}
\providecommand{\BIBentryALTinterwordspacing}{\spaceskip=\fontdimen2\font plus
\BIBentryALTinterwordstretchfactor\fontdimen3\font minus
  \fontdimen4\font\relax}
\providecommand{\BIBforeignlanguage}[2]{{%
\expandafter\ifx\csname l@#1\endcsname\relax
\typeout{** WARNING: IEEEtran.bst: No hyphenation pattern has been}%
\typeout{** loaded for the language `#1'. Using the pattern for}%
\typeout{** the default language instead.}%
\else
\language=\csname l@#1\endcsname
\fi
#2}}
\providecommand{\BIBdecl}{\relax}
\BIBdecl

\bibitem{Clarke81ctl}
E.~Clarke and E.~Emerson, ``Design and synthesis of synchronization skeletons
  using branching time temporal logic,'' in \emph{Proc. of LP'81}, ser. LNCS,
  vol. 131.\hskip 1em plus 0.5em minus 0.4em\relax Springer, 1981, pp. 52--71.

\bibitem{Queille81verification}
J.~Queille and J.~Sifakis, ``Specification and verification of concurrent
  programs in {C}esar,'' in \emph{SP'81}, ser. LNCS, vol. 137.\hskip 1em plus
  0.5em minus 0.4em\relax Springer, 1981, pp. 337--351.

\bibitem{Pnueli77}
A.~Pnueli, ``The temporal logic of programs,'' in \emph{Proc. 18th FOCS}.\hskip
  1em plus 0.5em minus 0.4em\relax IEEE Computer Society, 1977, pp. 46--57.

\bibitem{EmersonH86}
E.~Emerson and J.~Halpern, ``"{S}ometimes" and "{N}ot {N}ever" revisited: on
  branching versus linear time temporal logic,'' \emph{J. ACM}, vol.~33, no.~1,
  pp. 151--178, 1986.

\bibitem{ClarksonS10}
M.~Clarkson and F.~Schneider, ``Hyperproperties,'' \emph{Journal of Computer
  Security}, vol.~18, no.~6, pp. 1157--1210, 2010.

\bibitem{goguen1982security}
J.~Goguen and J.~Meseguer, ``Security policies and security models,'' in
  \emph{IEEE Symposium on Security and privacy}, vol.~12, 1982.

\bibitem{McLean96}
J.~McLean, ``A general theory of composition for a class of "possibilistic''
  properties,'' \emph{{IEEE} Trans. Software Eng.}, vol.~22, no.~1, pp. 53--67,
  1996.

\bibitem{ZdancewicM03}
S.~Zdancewic and A.~Myers, ``Observational determinism for concurrent program
  security,'' in \emph{Proc. 16th {IEEE} CSFW-16}.\hskip 1em plus 0.5em minus
  0.4em\relax {IEEE} Computer Society, 2003, pp. 29--43.

\bibitem{FinkbeinerRS15}
B.~Finkbeiner, M.~N. Rabe, and C.~S{\'{a}}nchez, ``Algorithms for model
  checking {HyperLTL} and {HyperCTL*},'' in \emph{Proc. 27th {CAV} Part {I}},
  ser. LNCS, vol. 9206.\hskip 1em plus 0.5em minus 0.4em\relax Springer, 2015,
  pp. 30--48.

\bibitem{DimitrovaFKRS12}
R.~Dimitrova, B.~Finkbeiner, M.~Kov{\'a}cs, M.~Rabe, and H.~Seidl, ``Model
  checking information flow in reactive systems,'' in \emph{Proc. 13th VMCAI},
  ser. LNCS 7148.\hskip 1em plus 0.5em minus 0.4em\relax Springer, 2012, pp.
  169--185.

\bibitem{ClarksonFKMRS14}
M.~Clarkson, B.~Finkbeiner, M.~Koleini, K.~Micinski, M.~Rabe, and
  C.~S{\'a}nchez, ``Temporal logics for hyperproperties,'' in \emph{Proc. 3rd
  POST}, ser. LNCS, vol. 8414.\hskip 1em plus 0.5em minus 0.4em\relax Springer,
  2014, pp. 265--284.

\bibitem{BozzelliMP15}
L.~Bozzelli, B.~Maubert, and S.~Pinchinat, ``Unifying {H}yper and {E}pistemic
  {T}emporal {L}ogics,'' in \emph{Proc. 18th FoSSaCS}, ser. LNCS, vol.
  9034.\hskip 1em plus 0.5em minus 0.4em\relax Springer, 2015, pp. 167--182.

\bibitem{Rabe2016}
M.~Rabe, ``A temporal logic approach to information-flow control,'' Ph.D.
  dissertation, Saarland University, 2016.

\bibitem{FinkbeinerH16}
B.~Finkbeiner and C.~Hahn, ``Deciding hyperproperties,'' in \emph{Proc. 27th
  {CONCUR}}, ser. LIPIcs, vol.~59.\hskip 1em plus 0.5em minus 0.4em\relax
  Schloss Dagstuhl - Leibniz-Zentrum f{\"{u}}r Informatik, 2016, pp.
  13:1--13:14.

\bibitem{CoenenFHH19}
N.~Coenen, B.~Finkbeiner, C.~Hahn, and J.~Hofmann, ``The hierarchy of
  hyperlogics,'' in \emph{Proc. 34th {LICS}}.\hskip 1em plus 0.5em minus
  0.4em\relax {IEEE}, 2019, pp. 1--13.

\bibitem{GutsfeldMO20}
J.~Gutsfeld, M.~M{\"{u}}ller{-}Olm, and C.~Ohrem, ``Propositional dynamic logic
  for hyperproperties,'' in \emph{Proc. 31st {CONCUR}}, ser. LIPIcs 171.\hskip
  1em plus 0.5em minus 0.4em\relax Schloss Dagstuhl - Leibniz-Zentrum f{\"{u}}r
  Informatik, 2020, pp. 50:1--50:22.

\bibitem{SistlaVW87}
A.~Sistla, M.~Vardi, and P.~Wolper, ``The complementation problem for
  {B}{\"u}chi automata with appplications to temporal logic,''
  \emph{Theoretical Computer Science}, vol.~49, pp. 217--237, 1987.

\bibitem{FischerL79}
M.~Fischer and R.~Ladner, ``Propositional dynamic logic of regular programs,''
  \emph{J. Comput. Syst. Sci.}, vol.~18, no.~2, pp. 194--211, 1979.

\bibitem{Finkbeiner17}
B.~Finkbeiner, ``Temporal hyperproperties,'' \emph{Bull. {EATCS}}, vol. 123,
  2017.

\bibitem{GutsfeldOO21}
J.~Gutsfeld, M.~M{\"{u}}ller{-}Olm, and C.~Ohrem, ``Automata and fixpoints for
  asynchronous hyperproperties,'' \emph{Proc. {ACM} Program. Lang.}, vol.~4,
  no. {POPL}, 2021.

\bibitem{CoenenFS21}
J.~Baumeister, N.~Coenen, B.~Bonakdarpour, B.~Finkbeiner, and C.~S{\'{a}}nchez,
  ``A temporal logic for asynchronous hyperproperties,'' in \emph{Proc. of 33rd
  CAV'21}, ser. LNCS, vol. 12759.\hskip 1em plus 0.5em minus 0.4em\relax
  Springer, 2021.

\bibitem{Finkbeiner017}
B.~Finkbeiner and M.~Zimmermann, ``The first-order logic of hyperproperties,''
  in \emph{Proc. 34th {STACS}}, ser. LIPIcs, vol.~66.\hskip 1em plus 0.5em
  minus 0.4em\relax Schloss Dagstuhl - Leibniz-Zentrum f{\"{u}}r Informatik,
  2017, pp. 30:1--30:14.

\bibitem{HU79}
J.~Hopcroft and J.~Ullman, \emph{Introduction to Automata Theory, Languages and
  Computation}.\hskip 1em plus 0.5em minus 0.4em\relax Addison-Wesley, 1979.

\bibitem{VardiW94}
M.~Y. Vardi and P.~Wolper, ``Reasoning about infinite computations,''
  \emph{Inf. Comput.}, vol. 115, no.~1, pp. 1--37, 1994.

\bibitem{Minsky67}
M.~L. Minsky, \emph{Computation: Finite and Infinite Machines}, ser. Automatic
  Computation.\hskip 1em plus 0.5em minus 0.4em\relax Prentice-Hall, Inc.,
  1967.

\end{thebibliography}

\newpage
\onecolumn

\newenvironment{changemargin}{%
  \begin{list}{}{%
    \setlength{\leftmargin}{40bp}%
    \setlength{\rightmargin}{40bp}%
    \setlength{\textheight}{610bp}%
   \setlength{\topmargin}{-30bp}
  }%
  \item[]}{\end{list}}

\makeatletter
\renewcommand{\normalsize}{%
\@setfontsize\normalsize {11pt}{14pt}
   \normalbaselineskip=13pt
   }

\makeatother

\makeatletter\makeatother

\normalsize

\begin{changemargin}

\newcounter{aux}
\newcounter{auxSec}

\begin{center}
\begin{LARGE}
  \noindent\textbf{Appendix}
\end{LARGE}
\end{center}

\section{Proofs from Section~\ref{sec:StutteringHLTL}}

\subsection{\textbf{Proof of Proposition~\ref{prop:FromSLTLtoLTL}}}\label{APP:FromSLTLtoLTL}

\setcounter{aux}{\value{proposition}}
\setcounter{auxSec}{\value{section}}
\setcounter{section}{\value{sec-FromSLTLtoLTL}}
\setcounter{proposition}{\value{prop-FromSLTLtoLTL}}

\begin{proposition} Given a $\SLTL$ formula, one can construct in polynomial time an equivalent $\LTL$ formula.
\end{proposition}
\setcounter{proposition}{\value{aux}}
\setcounter{section}{\value{auxSec}}
\begin{proof}
Given an $\SLTL$ formula $\psi$, we construct in polynomial time an $\LTL$ formula $f(\psi)$ such that
for each pointed trace $(\pi,i)$, it holds that $(\pi,i)\models \psi$ iff $(\pi,i)\models f(\psi)$.
The proof is by induction on structure of $\psi$.
\begin{itemize}
  \item $\psi = p$ for some $p\in \AP$: we set $f(p)\DefinedAs p$ and the result trivially follows.\vspace{0.1cm}

  \item $\psi = \neg \psi'$ (resp., $\psi = \psi_1\wedge \psi_2$): we set $f(\psi)\DefinedAs \neg f(\psi')$ (resp., $f(\psi)\DefinedAs f(\psi_1)\wedge f(\psi_2)$) and the result directly follows from the induction hypothesis.\vspace{0.1cm}

  \item $\psi=\Next_\Gamma \psi'$: if $\Gamma=\emptyset$, we set $f(\psi)= \Next f(\psi')$ and the result directly follows from the induction hypothesis.
  Now, assume that $\Gamma\neq \emptyset$. We exploit the auxiliary $\LTL$ formula $\theta_\Gamma$ defined as: $\theta_\Gamma \DefinedAs \displaystyle{\bigwedge_{\xi\in\Gamma}}(\xi \leftrightarrow \Next \xi)$. Formula $\theta_\Gamma$ asserts that the given position and the next one along the given trace agree on the evaluation of all the $\LTL$ formulas in $\Gamma$. Hence, the $\LTL$ formula $\Always \theta_\Gamma$ requires that
  for the  given position $i$ of the given trace $\pi$, the number of segments in the $\Gamma$-stutter factorization of $\pi$ is finite and $i$ belongs to the last segment. Thus, the $\LTL$ formula $f(\Next_\Gamma \psi')$ is defined as follows:
  \[
     f(\Next_\Gamma \psi')\DefinedAs [\Always\theta_\Gamma \rightarrow \Next f(\psi')] \wedge [ \neg\Always\theta_\Gamma \rightarrow  \theta_\Gamma \Until (\neg \theta_\Gamma \wedge \Next f(\psi'))]
  \]
  Given a pointed trace $(\pi,i)$, $f(\Next_\Gamma \psi')$ requires that either (i) position $i$ belongs to the last segment of the
   $\Gamma$-stutter factorization of $\pi$ and $f(\psi')$ holds at the next position $i+1$, or (ii) the $i$-segment in the
    $\Gamma$-stutter factorization of $\pi$ is not the last one and $f(\psi')$ holds at the first position of the first segment following the $i$-segment. Hence, by the induction hypothesis, correctness of the construction follows.\vspace{0.1cm}

    \item $\psi=\psi_1\Until_\Gamma \psi_2$: if $\Gamma=\emptyset$, we set $f(\psi)=f(\psi_1) \Until f(\psi_2)$ and the result directly follows from the induction hypothesis.
  Now, assume that $\Gamma\neq \emptyset$. In this case, the $\LTL$ formula $f(\psi_1\Until_\Gamma \psi_2)$ is defined as follows:
  \[ \begin{array}{ll}
 f(\psi_1\Until_\Gamma \psi_2) \DefinedAs  &    [\Always\theta_\Gamma \rightarrow  f(\psi_1)\Until f(\psi_2)] \wedge\\
   &  [ \neg\Always\theta_\Gamma \rightarrow [f(\psi_2) \vee (f(\psi_1)\wedge (\neg\theta_\Gamma \rightarrow \Next f(\psi_1)) \,\Until\, g(\psi_1\Until_\Gamma \psi_2))]]\\
 g(\psi_1\Until_\Gamma \psi_2)\DefinedAs & \neg \theta_\Gamma \wedge \Next [ f(\psi_2) \vee (\Always\theta_\Gamma \wedge f(\psi_1)\Until f(\psi_2))]
\end{array} \]
  Given a pointed trace $(\pi,i)$, $f(\psi_1\Until_\Gamma \psi_2)$ requires that either (i) position $i$ belongs to the last segment of the
   $\Gamma$-stutter factorization of $\pi$ and $f(\psi_1)\Next f(\psi_2)$ holds at  position $i$, or (ii) the $i$-segment in the
    $\Gamma$-stutter factorization of $\pi$ is not the last one and
    \begin{compactitem}
      \item \emph{either} $f(\psi_2)$ holds at position $i$,
      \item \emph{or} $f(\psi_1)$ holds at position $i$ and there is a segment $\nu$ following the $i$-segment  such that $f(\psi_2)$ holds at the first position of $\nu$ and $f(\psi_1)$ holds at the first position of each segment preceding $\nu$ end following the $i$-segment.
      \item \emph{or} $f(\psi_1)$ holds at position $i$,  the last  segment, say $\nu$, is defined,   $f(\psi_2)\Until f(\psi_2)$ holds at the first position of $\nu$, and $f(\psi_1)$ holds at the first position of each segment preceding $\nu$ end following the $i$-segment.
    \end{compactitem}
\end{itemize}
\noindent Note that the size of $f(\psi)$ (i.e., the number of distinct sub-formulas of $f(\psi)$) is polynomial in the size of $\psi$.
\end{proof}

\bigskip

\subsection{\textbf{Detailed proof of Proposition~\ref{prop:FromStutteringHLTLtoAAWA}}}\label{APP:FromStutteringHLTLtoAAWA}

\setcounter{aux}{\value{proposition}}
\setcounter{auxSec}{\value{section}}
\setcounter{section}{\value{sec-FromStutteringHLTLtoAAWA}}
\setcounter{proposition}{\value{prop-FromStutteringHLTLtoAAWA}}

\begin{proposition}   Given a  $\SHLTL$ quantifier-free formula $\psi$ with trace variables $x_1,\ldots,x_n$,
one can build in polynomial time a B\"{u}chi  $\NAAWA$ $\Au_\psi$ such that $\Lang(\Au_\psi)$ is the set of $n$-tuples $(\pi_1,\ldots,\pi_n)$ of traces so that $(\{x_1\mapsto (\pi_1,0),\ldots,x_1\mapsto (\pi_n,0)\})\models \psi$.
\end{proposition}
\setcounter{proposition}{\value{aux}}
\setcounter{section}{\value{auxSec}}
\begin{proof}
Without loss of generality, we assume that  $\psi$ is in \emph{negative normal form} ($\textit{NNF}$), i.e. negation is applied only to relativized  atomic propositions.  Indeed, by exploiting the dual $\Release_\Gamma$ (\emph{relativized release}) of the until modality $\Until_\Gamma$, both conjunction and disjunction, and the De Morgan's laws,   we can convert in linear time a given $\SHLTL$ quantifier-free  formula $\theta$ into an equivalent
$\SHLTL$ quantifier-free  formula $\theta'$ in $\textit{NNF}$.

\noindent For each finite set  $\Gamma$ of $\LTL$ formulas, let $\xi_{\Gamma}$ be the $\LTL$ formula given by
\[
\xi_\Gamma = \displaystyle{\bigwedge_{\xi\in \Gamma} \Always(\xi \leftrightarrow \Next \xi) \vee \bigvee_{\xi\in\Gamma} (\xi \leftrightarrow \neg \Next\xi)}
\]
\noindent The $\LTL$ formula $\xi_\Gamma$ has as models the traces $\pi$ such that the first segment in the factorization of $\pi$ is either infinite or has length $1$. For each $i\in [1,n]$, we can easily construct in linear time (in the number of distinct sub-formulas in $\Gamma$) a
B\"{u}chi  $\NAAWA$ $\Au_{\Gamma,i}= \tpl{2^{\AP},Q_{\Gamma,i},q_{\Gamma,i},\rho_{\Gamma,i},F_{\Gamma,i}}$ (resp., $\overline{\Au}_{\Gamma,i}= \tpl{2^{\AP},\overline{Q}_{\Gamma,i},\overline{q}_{\Gamma,i},\overline{\rho}_{\Gamma,i},\overline{F}_{\Gamma,i}}$) accepting the
$n$-tuples $(\pi_1,\ldots,\pi_n)$ of traces so that the $i^{th}$ component $\pi_i$  is a model (resp., is not a model) of $\xi_\Gamma$.

Let $\Upsilon$ be the set of subscripts $\Gamma$ occurring in the temporal modalities of $\psi$.
 The B\"{u}chi  $\NAAWA$ $\Au_\psi=\tpl{2^{\AP},Q=Q_s\cup Q_m,q_0,\rho,F}$ satisfying Proposition~\ref{prop:FromStutteringHLTLtoAAWA}
is defined as follows, where we assume that the various automata $\Au_{\Gamma,i}$ and $\overline{\Au}_{\Gamma,i}$ have no state in common. The set $Q_s$  of states  is the  union of the states of the auxiliary automata $\Au_{\Gamma,i}$ and $\overline{\Au}_{\Gamma,i}$ where $\Gamma\in\Upsilon$ and $i\in [1,n]$, while the states in  $Q_m$ consist  of the sub-formulas of $\psi$ and the triples of the form $(\theta,i,\Gamma)$ where $\theta$ is a sub-formula of $\psi$,
$i\in [1,n]$ and $\Gamma\in\Upsilon$. The initial state is $\psi$ and the set of accepting states consists of the states of the form $\theta_1\Release \theta_2$
and the accepting states of the automata $\Au_{\Gamma,i}$ and $\overline{\Au}_{\Gamma,i}$ where $\Gamma\in\Upsilon$ and $i\in [1,n]$.

Finally, we define the transition function $\rho$. From the states in $Q_s$ the transition function is inherited from
the respective $\NAAWA$ $\Au_{\Gamma,i}$ and $\overline{\Au}_{\Gamma,i}$. For the states in $Q_m$, $\rho$ is defined by induction on the structure of the sub-formulas $\theta$ of $\psi$ as follows, where for each $\vect{\sigma}\in (2^{\AP})^{n}$ and $i\in [1,n]$, $\vect{\sigma}[i]$ denotes the $i^{th}$ component of $\vect{\sigma}$:
\begin{itemize}
  \item for each $i\in [1,n-1]$,  $\rho((\theta,i,\Gamma),\vect{\sigma})  = \bigl(\rho_{\Gamma,i}(q_{\Gamma,i},\vect{\sigma})\wedge ((\theta,i+1,\Gamma),i)\bigr)  \,\vee $

  \hspace{6.35cm} $\bigl(\overline{\rho}_{\Gamma,i}(\overline{q}_{\Gamma,i},\vect{\sigma})\wedge ((\theta,i,\Gamma),i)\bigr)$ \vspace{0.2cm}

 \item  $\rho((\theta,n,\Gamma),\vect{\sigma})  = \bigl(\rho_{\Gamma,n}(q_{\Gamma,n},\vect{\sigma})\wedge (\theta,n)\bigr)  \vee \bigl(\overline{\rho}_{\Gamma,n}(\overline{q}_{\Gamma,n},\vect{\sigma})\wedge ((\theta,n,\Gamma),n)\bigr)$ \vspace{0.2cm}

\item
  $\rho(\Rel{p}{x_i},\vect{\sigma})    =  \left\{
  \begin{array}{ll}
 \true
             &    \text{ if }  p\in\vect{\sigma}[i]
             \\
 \false            &    \text{ otherwise }
\end{array}
             \right.
$ \vspace{0.2cm}

\item
  $
  \rho(\neg \Rel{p}{x_i},\vect{\sigma})    =  \left\{
  \begin{array}{ll}
 \true
             &      \text{ if }  p\notin\vect{\sigma}[i]
             \\
 \false            &    \text{ otherwise }
\end{array}
             \right.
$ \vspace{0.2cm}

\item
  $\rho(\theta_1\vee \theta_2,\vect{\sigma})    =
\rho(\theta_1,\vect{\sigma})\vee \rho(\theta_2,\vect{\sigma})$ \vspace{0.2cm}

\item
  $\rho(\theta_1\wedge \theta_2,\vect{\sigma})    =
\rho(\theta_1,\vect{\sigma})\vee \rho(\theta_2,\vect{\sigma})$ \vspace{0.2cm}

\item
  $
  \rho(\Next_\Gamma \theta,\vect{\sigma})    = \rho((\theta,1,\Gamma),\vect{\sigma})
$ \vspace{0.2cm}

\item
  $
  \rho(\theta_1\Until_{\Gamma} \theta_2\vect{\sigma})    =
  \rho(\theta_2,\vect{\sigma})  \vee  \bigl(\rho(\theta_1,\vect{\sigma}) \wedge \rho((\theta_1\Until_{\Gamma} \theta_2,1,\Gamma),\vect{\sigma})\bigr)
$ \vspace{0.2cm}

\item
 $
  \rho(\theta_1\Release_{\Gamma} \theta_2\vect{\sigma})    =
  \rho(\theta_2,\vect{\sigma})  \wedge  \bigl(\rho(\theta_1,\vect{\sigma}) \vee \rho((\theta_1\Release_{\Gamma} \theta_2,1,\Gamma),\vect{\sigma})\bigr)
$
\end{itemize}

For a pointed trace assignment $\Pi= \{x_1 \mapsto (\pi_1,j_1),\ldots,x_n \mapsto (\pi_n,j_n)\}$ over
$\{x_1,\ldots,x_n\}$, we denote by $\vect{\pi}(\Pi)$ the $n$-tuple of traces $(\pi_1,\ldots,\pi_n)$ and by
$\wp(\Pi)$ the position vector $(j_1,\ldots,j_n)$.

Given an $n$-tuple $\vect{\pi}=(\pi_1,\ldots,\pi_n)$ of traces, a position vector $\wp=(j_1,\ldots,j_n)$, and a state $q$ of $\Au_\psi$,
a \emph{$(q,\wp)$-run  of $\Au_\psi$ over the input $\vect{\pi}$} is defined as a run of $\Au_\psi$ over
$\vect{\pi}$ but we require that the root of the run is labeled by the pair $(q,\wp)$ (i.e., initially, the automaton is in state $q$ and for each $i\in [1,n]$ reads the $(j_i+1)^{th}$ symbol of the trace $\pi_i$).
By construction and a straightforward induction on the structure of the sub-formulas of $\psi$, we obtain the following result, hence,
Proposition~\ref{prop:FromStutteringHLTLtoAAWA} directly follows.\vspace{0.1cm}

\noindent \textbf{Claim.} Let $\Pi$ be a pointed trace assignment over $\{x_1,\ldots,x_n\}$ and $\theta$ be a sub-formula of $\psi$. Then, there is an
accepting $(\theta,\wp(\Pi))$-run of $\Au_\psi$ over $\vect{\pi}(\Pi)$ \emph{if and only if} $\Pi\models \theta$.
\end{proof}

\bigskip

\subsection{\textbf{Proof of Proposition~\ref{prop:KripkeBuild}}}\label{APP:KripkeBuild}

\setcounter{aux}{\value{proposition}}
\setcounter{auxSec}{\value{section}}
\setcounter{section}{\value{sec-KripkeBuild}}
\setcounter{proposition}{\value{prop-KripkeBuild}}

   \begin{proposition}  One can build in time polynomial in the size of $\Instance$ a finite Kripke structure
 $\Ku_\Instance$ over $\AP$ satisfying the following conditions:
   \begin{compactitem}
     \item the set of  traces of $\Ku_\Instance$ contains the set of well-formed traces;
     \item each  trace of $\Ku_\Instance$ having a suffix where $\#$ always holds is a well-formed trace.
   \end{compactitem}
   \end{proposition}
   \setcounter{proposition}{\value{aux}}
\setcounter{section}{\value{auxSec}}
   \begin{proof} Essentially, at the unique initial state, $\Ku_\Instance$ non-deterministically chooses to generate symbol by symbol  either a well-formed
   trace associated to the first tuple of $\Instance$ ($\ell=1$) or a well-formed trace associated
   to the second tuple of $\Instance$ ($\ell=2$). On generating (symbol by symbol) an infix $[u_{i}^{\ell},p_{i},q_\ell]$ with $i\in [n]$ and $\ell=1,2$, $\Ku_\Instance$ keeps track in its state of the word $u_{i}^{\ell}$, of $\ell=1,2$, and of the position $0\leq h< |u_{i}^{\ell}|$ of the symbol in $2^{\AP}$ currently generated along the infix. On generating the last symbol of this infix, $\Ku_\Instance$ non-deterministically chooses either  to generate symbol by symbol another infix of the form $[u_{j}^{\ell},p_{j},q_\ell]$ for some $j\in [n]$ or to move to a trap state $s_f$ from which  the suffix $\{\#\}^{\omega}$ is generated.\vspace{0.1cm}

  \noindent Formally, $\Ku_\Instance= \tpl{S,\{s_0\},E,V}$ where
   \begin{compactitem}
     \item $S\DefinedAs \{s_0,s_f\}\cup \{(u_i^{\ell},\ell,h)\mid i\in [n],\,\ell=1,2,\,0\leq h< |u_i^{\ell}| \}$;
     \item $E$ consist of the following transitions:
     \begin{compactitem}
     \item $(s_0,(u_i^{\ell},\ell,0))$ for all $i\in [n]$ and $\ell=1,2$;
     \item $((u_i^{\ell},\ell,h),(u_i^{\ell},\ell,h+1))$ for all $i\in [n]$, $\ell=1,2$, and $0\leq h<|u_i^{\ell}|-1$;
     \item $((u_i^{\ell},\ell,|u_i^{\ell}|-1),(u_j^{\ell},\ell,0))$ for all $i,j\in [n]$ and $\ell=1,2$;
    \item $((u_i^{\ell},\ell,|u_i^{\ell}|-1),s_f)$ for all $i\in [n]$ and $\ell=1,2$;
    \item $(s_f,s_f)$.
   \end{compactitem}
   \item The $\AP$-valuation $V$ is defined as follows:
      \begin{compactitem}
     \item $V(s_0)=V(s_f)=\{\#\}$;
     \item $V((u_i^{\ell},\ell,h))=\{u_i^{\ell}(h),p_i,q_\ell\}$ for all $i\in [n]$, $\ell=1,2$, and $0\leq h<|u_i^{\ell}|-1$;
     \item $V((u_i^{\ell},\ell,h)=\{u_i^{\ell}(h),p_i,q_\ell,\#\}$ for all $i\in [n]$, $\ell=1,2$, and $h= |u_i^{\ell}|-1$.
   \end{compactitem}
   \end{compactitem}
   \end{proof}

\bigskip

\subsection{\textbf{Proof of Proposition~\ref{prop:NFAforLTL}}}\label{APP:NFAforLTL}

\setcounter{aux}{\value{proposition}}
\setcounter{auxSec}{\value{section}}
\setcounter{section}{\value{sec-NFAforLTL}}
\setcounter{proposition}{\value{prop-NFAforLTL}}

\begin{proposition}  Given a finite set $\Gamma$ of $\LTL$ formulas over $\AP$, one can construct in singly exponential time a
$\NBA$ $\Au_\Gamma$ over $2^{\AP_\Gamma}$ with $2^{O(|\AP_\Gamma|)}$ states satisfying the following:
\begin{compactenum}
  \item let $w\in\Lang(\Au_\Gamma)$: then for all $i\geq 0$ and $\theta\in \cl(\Gamma)$, $at(\theta)\in w(i)$ if and only if $((w)_{\AP},i)\models \theta$.
  \item for each trace $\pi$ (i.e., infinite word over $2^{\AP}$), there exists $w\in \Lang(\Au_\Gamma)$ such that $\pi= (w)_{\AP}$.
\end{compactenum}
\end{proposition}
\setcounter{proposition}{\value{aux}}
\setcounter{section}{\value{auxSec}}
\begin{proof} Here, we construct a \emph{generalized} $\NBA$
  $\Au_\Gamma$ satisfying Properties~(1) and~(2) of
  Proposition~\ref{prop:NFAforLTL} which can be converted in linear
  time into an equivalent $\NBA$. Recall that a generalized
  $\NBA$ is defined as a $\NBA$ but the acceptance condition
  is given by a family $\mathcal{F}=\{\Acc_1,\ldots,\Acc_k\}$ of sets
  of accepting states. In this case, a run is accepting if for each
  accepting component $\Acc_i\in \mathcal{F}$, the run visits
  infinitely often states in $\Acc_i$.

The generalized $\NBA$ $\Au_\Gamma = \tpl{2^{\AP_\Gamma},Q,Q_0,\Delta,\mathcal{F}}$  is defined as follows.
$Q$ is the set of \emph{atoms} of $\Gamma$ consisting of the maximal  propositionally consistent subsets $A$ of $\cl(\Gamma)$.
Formally, an atom $A$ of $\Gamma$ is a subset of $\cl(\Gamma)$  satisfying the following:
 \begin{compactitem}
  \item for each $\theta\in\cl(\Gamma)$, $\theta\in A$ iff $\neg\theta\notin A$;
    \item for each $\theta_1\wedge \theta_2\in \cl(\Gamma)$, $\theta_1\wedge\theta_2\in A$ iff $\theta_1,\theta_2\in A$.
\end{compactitem}
Each state is initial, i.e. $Q_0= Q$.
For an atom $A$, $at(A)$ denotes the subset of propositions in $\AP_\Gamma$ associated with the formulas in $A$, i.e.
$at(A)\DefinedAs\{at(\theta)\mid \theta\in A\}$.

\noindent The transition relation $\Delta$ captures the semantics of the next modality, and the local fixpoint characterization of the until  modality.
Formally, $\Delta$ consists of the transitions of the form $(A,at(A),A')$ such that:
 \begin{compactitem}
 \item for each $\Next\theta\in\cl(\Gamma)$, $\Next\theta\in A$ iff $\theta\in A'$;
  \item for each $\theta_1\Until\theta_2\in\cl(\Gamma)$, $\theta_1\Until\theta_2\in A$ iff either $\theta_2\in A$, or $\theta_1\in A$  and $\theta_1\Until\theta_2\in A'$.
\end{compactitem}
Finally, the generalized B\"{u}chi acceptance condition is used for ensuring the fulfillment of the liveness requirements $\theta_2$ in the until sub-formulas $\theta_1\Until\theta_2$ in $\Gamma$. Formally, for each $\theta_1\Until\theta_2\in\cl(\Gamma)$, $\mathcal{F}$ has a component
consisting of the atoms $A$ such that either $\neg(\theta_1\Until\theta_2)\in A$ or $\theta_2\in A$.

Let $w\in\Lang(\Au_\Gamma)$. By construction, there is an accepting infinite sequence of atoms $\rho=A_0 A_1\ldots$  such that for all $i\geq 0$, $w(i)=at(A_i)$. Let $\pi$ be the projection of $w$ over $\AP$ (note that $A_i\cap\AP = \pi(i)$ for all $i\geq 0$).
By standard arguments (see~\cite{VardiW94}), the following holds: for all $i\geq 0$ and $\theta\in \cl(\Gamma)$, $\theta\in A_i$ (hence, $at(\theta)\in w(i)$) if and only if $(\pi,i)\models \theta$.
Hence, Property~(1) of Proposition~\ref{prop:NFAforLTL} follows.

For Property~(2), let $\pi$ be a trace (i.e., an infinite word over $2^{\AP}$) and let $\rho= A_0 A_1\ldots$ be the infinite sequence of atoms defined
as follows for all $i\geq 0$: $A_i=\{\theta\in \cl(\Gamma)\mid (\pi,i)\models \theta\}$. By construction and the semantics of $\LTL$,
$\rho$ is an accepting run of $\Au_\Gamma$ over the word $w= at(A_0) at(A_1)\ldots$. Moreover, $\pi$ coincides with the projection of $w$ over $\AP$. Hence, the result follows.
\end{proof}

\bigskip

\subsection{\textbf{Proof of Proposition~\ref{prop:conditionedLTL}}}\label{APP:conditionedLTL}

\setcounter{aux}{\value{proposition}}
\setcounter{auxSec}{\value{section}}
\setcounter{section}{\value{sec-conditionedLTL}}
\setcounter{proposition}{\value{prop-conditionedLTL}}

 \begin{proposition}  Given an $\LTL$ formula $\theta$ and a $\HLTL$ sentence $\varphi$, one can construct
 in linear time a  $\HLTL$ sentence $\varphi_\theta$ having the same quantifier prefix as $\varphi$ such that for each Kripke structure $\Ku$,
 $\Lang(\Ku,\theta) \models \varphi$ \emph{if and only if} $\Lang(\Ku)\models \varphi_\theta$.
 \end{proposition}
\setcounter{proposition}{\value{aux}}
\setcounter{section}{\value{auxSec}}
 \begin{proof}
 Let $\theta$ be a $\LTL$ formula and $\varphi$ be a  $\HLTL$ sentence of the form
 \[
 Q_n x_n.\,\ldots Q_1 x_1.\, \psi
 \]
 where $n\geq 1$,  $Q_i\in \{\exists,\forall\}$ for all $i\in [1,n]$, and $\psi$ is quantifier-free.
 Since $\varphi$ is a sentence, we can assume that $\Var=\{x_1,\ldots,x_n\}$ where variables $x_1,\ldots,x_n$ are pairwise distinct.
 For each $i\in [1,n]$, we denote by $\theta(x_i)$ the $\HLTL$ formula obtained from $\theta$ by replacing each occurrence of an atomic proposition
 $p$ in $\theta$ with its $\HLTL$ version $\Rel{p}{x_i}$ (referring to the trace associated with variable $x_i$).

 Let $\psi_{0},\, \psi_1\,\ldots\,\psi_n$ be the $\HLTL$ quantifier-free formulas inductively defined as follows:
 $\psi_{0}\DefinedAs\psi$ and for all $i\in [1,n]$:
  \[
  \begin{array}{l}
    \psi_i  =  \left\{
      \begin{array}{ll}
         \theta(x_i)\wedge \psi_{i-1}
        &    \text{ if }  Q_i=\exists
        \\
       \theta(x_i) \rightarrow \psi_{i-1}
        &    \text{ otherwise }
      \end{array}
    \right.
  \end{array}
  \]
 Define $\varphi_\theta\DefinedAs   Q_n x_n.\,\ldots Q_1 x_1.\, \psi_n$.
 We show that for each Kripke structure $\Ku$,  $\Lang(\Ku,\theta) \models \varphi$ iff $\Lang(\Ku)\models \varphi_\theta$. Hence, the result follows.
 Since $\varphi_\theta$ is a   $\HLTL$ sentence, it suffices to show the following for a given Kripke structure $\Ku$: \vspace{0.1cm}

 \noindent \emph{Claim.}  Let $\Pi$ be a pointed trace assignment such that \emph{either} $\Lang(\Ku,\theta)=\emptyset$ and $\Pi$ is empty,
 \emph{or} $\Lang(\Ku,\theta)\neq\emptyset$ and
$\Pi$ associates to each variable in $\Var$ a pointed trace $(\pi,0)$ such that $\pi\in\Lang(\Ku,\theta)$. Then,  for each $i\in [1,n]$, $(\Lang(\Ku,\theta),\Pi)\models  Q_i x_i.\,\ldots Q_1 x_1.\, \psi$ iff
 $(\Lang(\Ku),\Pi)\models  Q_i x_i.\,\ldots Q_1 x_1.\, \psi_i$. \vspace{0.1cm}

\noindent \emph{Proof of the claim.} The proof is by induction on $i\in [1,n]$. We focus on the induction step (the base case where $i=1$ being similar).
Let $i\in [2,n]$. We assume that $Q_i=\forall$ and $\Lang(\Ku,\theta)\neq\emptyset$ (the case where $Q_i=\exists$ or $\Lang(\Ku,\theta)=\emptyset$ being similar). Then,
  \[
  \begin{array}{l}
    (\Lang(\Ku,\theta),\Pi)\models  Q_i x_i.\,\ldots Q_1 x_1.\, \psi  \Leftrightarrow \\
    \text{for each } \pi\in\Lang(\Ku,\theta),\,
 (\Lang(\Ku,\theta),\Pi[x_i \mapsto (\pi,0)])\models  Q_{i-1} x_{i-1}.\,\ldots Q_1 x_1.\, \psi\Leftrightarrow  \\
 \text{for each } \pi\in\Lang(\Ku,\theta),\,
  (\Lang(\Ku),\Pi[x_i \mapsto (\pi,0)])\models  Q_{i-1} x_{i-1}.\,\ldots Q_1 x_1.\, \psi_{i-1}\Leftrightarrow  \\
 \text{for each } \pi\in\Lang(\Ku),\,
  (\Lang(\Ku),\Pi[x_i\mapsto (\pi,0)])\models  Q_{i-1} x_{i-1}.\,\ldots Q_1 x_1.\, \psi_{i}\Leftrightarrow  \\
  (\Lang(\Ku),\Pi)\models  Q_{i} x_{i}.\,\ldots Q_1 x_1.\, \psi_{i}
  \end{array}
  \]
 The second equivalence directly follows from the induction hypothesis, while the third equivalence follows from the semantics of $\HLTL$ and the facts
  that $x_1,\ldots,x_n$ are pairwise distinct and $\psi_i=\theta(x_i) \rightarrow \psi_{i-1}$ when $Q_i=\forall$.
 \end{proof}

\newpage

\section{Proofs from Section~\ref{sec:ContextHyper}}

\subsection{\textbf{Proof of Theorem~\ref{theo:UndecidabilityContextHyeprLTL} for the fragment $\F_1$}}\label{APP:UndecidabilityContextHyeprLTL}

We show that model checking against the fragment $\F_1$ of $\CHLTL$ is undecidable. For the fixed Minsky $2$-counter machine $M$
(see Subsection~\ref{sec:UndeciablilityContextHLTL}), the result directly follows from the following proposition.

\begin{proposition}
One can build 
a finite Kripke Structure $\Ku_M$ and a $\CHLTL$ sentence $\varphi_M$ in the fragment $\F_1$
such that $M$ halts \emph{iff} $\Ku_M\models \varphi_M$.
\end{proposition}
\begin{proof}
The proof is an adaptation of the proof of Proposition~\ref{prop:undecidabilityContextHLTL}. We consider the set
  of atomic propositions exploited in the proof of Proposition~\ref{prop:undecidabilityContextHLTL} augmented with two additional propositions $\#$ and $\#'$ used to mark $c$-codes for $c\in\{1,2\}$. Hence,
\[
\AP\DefinedAs \Delta \cup \{1,2,\Beg_1,\Beg_2,\#,\#'\}
\]
The notion of \emph{unmarked computation}-code corresponds to the notion of com\-pu\-ta\-tion-code in the proof of Proposition~\ref{prop:undecidabilityContextHLTL},
while a \emph{marked computation}-code $\pi$ corresponds to a computation-code  but we, additionally, require that there are exactly two $c$-codes $\pi_c$ and $\pi'_c$ along $\pi$ for some $c\in\{1,2\}$ marked by $\#$ and $\#'$, respectively, such that $\pi_c$ and $\pi'_c$ belongs to two adjacent configuration codes along $\pi$ with $\pi_c$ preceding $\pi'_c$.

 By construction, it is a trivial task to define a Kripke structure $\Ku_M$  satisfying the following.\vspace{0.2cm}

\noindent \textbf{Claim.} One can construct in polynomial time a finite Kripke structure $\Ku_M$ over $\AP$ such that the set of   traces of $\Ku_M$
which visit some empty position (i.e., a position with label the empty set of propositions) corresponds to the set of initial and halting (marked and unmarked) computation-codes.\vspace{0.2cm}

The $\CHLTL$ formula $\varphi_M$ in the fragment $\F_1$ is defined as follows:
\[
\begin{array}{l}
\varphi_M\DefinedAs \exists x_1.\,\forall x_2.\, \Always \neg \#[x_1]\wedge \Eventually\displaystyle{\bigwedge_{p\in\AP}} \neg \Rel{p}{x_1} \,\wedge \\
 \Bigl( \Eventually \#[x_2]\wedge \Eventually\displaystyle{\bigwedge_{p\in\AP}} \neg \Rel{p}{x_2} \wedge \Always    \displaystyle{\bigwedge_{p\in\AP\setminus\{\#,\#'\}}}(\Rel{p}{x_1}\leftrightarrow \Rel{p}{x_2})\Bigr) \longrightarrow \psi'_{good}
\end{array}
\]
 where the $\CHLTL$  quantifier-free  sub-formula $\psi'_{good}$ is defined in the following. Intuitively, when interpreted on the Kripke structure
 $\Ku_M$ of the previous claim, formula $\varphi_M$ asserts the existence of a  trace $\pi_1$ (bounded to the trace variable $x_1$) such that
 $\pi_1$ is an unmarked initial and halting computation-code and for each  trace $\pi_2$ (bounded to the trace variable $x_2$) so that
 $\pi_2$ is some marked version of $\pi_1$, $\pi_2$ satisfies (by means of the conjunct $\psi'_{good}$) the
 goodness requirement applied to the marked $c$-codes of $\pi_2$ for some $c\in\{1,2\}$. Since variable $x_2$ is universally quantified, this ensures that
 $\pi_1$ is also good.

 The quantifier-free $\CHLTL$ formula $\psi'_{good}$ is defined as follows:
\[
\begin{array}{l}
\psi'_{good}\DefinedAs   \Always\displaystyle{\bigwedge_{\delta\in\Delta\setminus \Delta_\halt}\bigwedge_{c\in\{1,2\}}}\Bigl[(\delta[x_1]\wedge \Beg_c[x_1]\wedge \#[x_2])\longrightarrow
 \\
 \hspace{1cm} \tpl{\{x_2\}}  \Eventually \bigl(\Beg_c[x_2] \wedge \#'[x_2] \wedge \tpl{\{x_1,x_2\}}(\psi_{=}(\delta,c)\wedge \psi_{inc}(\delta,c)\wedge \psi_{dec}(\delta,c))\bigr)\Bigr]
\end{array}
\]
where the sub-formulas  $\psi_{=}(\delta,c)$, $\psi_{inc}(\delta,c)$, and $\psi_{dec}(\delta,c)$ are defined as in the proof of Proposition~\ref{prop:undecidabilityContextHLTL}.
\end{proof}

\bigskip

\subsection{\textbf{Detailed proof of Proposition~\ref{prop:FromContextHyperToAAWA}}}\label{APP:FromContextHyperToAAWA}

\setcounter{aux}{\value{proposition}}
\setcounter{auxSec}{\value{section}}
\setcounter{section}{\value{sec-FromContextHyperToAAWA}}
\setcounter{proposition}{\value{prop-FromContextHyperToAAWA}}

\begin{proposition}  Given a  $\CHLTL$ quantifier-free formula $\psi$ with trace variables $x_1,\ldots,x_n$,
one can build in polynomial time a B\"{u}chi  $\NAAWA$ $\Au_\psi$ such that $\Lang(\Au_\psi)$ is the set of $n$-tuples $(\pi_1,\ldots,\pi_n)$ of traces so that $(\{x_1\mapsto (\pi_1,0),\ldots,x_1\mapsto (\pi_n,0)\},\{x_1,\ldots,x_n\})\models \psi$. Moreover, $\Au_\psi$ is
$(|\psi|+1)$-synchronous if $\psi$ is in the bounded fragment of $\CHLTL$.
\end{proposition}
\setcounter{proposition}{\value{aux}}
\setcounter{section}{\value{auxSec}}
\begin{proof}
By exploiting the release modality $\Release$ (the dual of the until modality), both conjunction and disjunction,  and the fact that a formula $\neg \tpl{C}\theta$ corresponds to $\tpl{C}\neg\theta$, we can assume without loss of generality that
 $\psi$ is in  negative normal form, i.e. negation is applied only to relativized  atomic propositions.

\noindent For each $i\in [1,n]$, we write $\dir(x_i)$ to mean the direction $i$. Moreover, for each context $C\subseteq \{x_1,\ldots,x_n\}$, we denote by
$\first(C)$ (resp., $\last(C)$) the smallest (resp., greatest) variable occurring in $C$ with respect to the fixed ordering $x_1<\ldots <x_n$. Additionally, for each $x\in C\setminus \{\last(C)\}$, we denote by $\SUCC(C,x)$  the smallest variable in $C$ which is greater than $x$.

\noindent The B\"{u}chi  $\NAAWA$ $\Au_\psi=\tpl{2^{\AP},Q,q_0,\rho,F}$ is defined as follows. The set $Q$ of states is the set of triples
$(\theta,C,x)$ such that $\theta$ is a sub-formula of $\psi$, $C$ is a context such that some occurrence of $\theta$ in $\psi$ is in the scope of the context
modality $\tpl{C}$, and $x\in C$. The initial state $q_0$ is given by $(\psi,\{x_1,\ldots,x_n\},x_n)$, while the set $F$ of accepting states consists of the states of the
form $(\theta_1\Release \theta_2, C,last(C))$.

Finally, the transition function $\rho$ is defined by induction on the structure of the sub-formulas $\theta$ of $\psi$ as follows, where for each $\vect{\sigma}\in (2^{\AP})^{n}$ and $i\in [1,n]$, $\vect{\sigma}[i]$ denotes the $i^{th}$ component of $\vect{\sigma}$:
\begin{itemize}
  \item for $x\neq last(C)$, $\rho((\theta,C,x),\vect{\sigma})  = ((\theta,C,\SUCC(C,x)),\dir(\SUCC(C,x))$\vspace{0.2cm}

\item
  $\rho((\Rel{p}{x},C,\last(C)),\vect{\sigma})    =  \left\{
  \begin{array}{ll}
 \true
             &    \text{ if }  p\in\vect{\sigma}[\dir(x)]
             \\
 \false            &    \text{ otherwise }
\end{array}
             \right.
$ \vspace{0.2cm}

\item
  $
  \rho((\neg \Rel{p}{x},C,\last(C)),\vect{\sigma})    =  \left\{
  \begin{array}{ll}
 \true
             &      \text{ if }  p\notin\vect{\sigma}[\dir(x)]
             \\
 \false            &    \text{ otherwise }
\end{array}
             \right.
$ \vspace{0.2cm}

\item
  $\rho((\theta_1\vee \theta_2,C,\last(C)),\vect{\sigma})    =
\rho((\theta_1,C,\last(C)),\vect{\sigma})\vee \rho((\theta_2,C,\last(C)),\vect{\sigma})$ \vspace{0.2cm}

\item
  $\rho((\theta_1\wedge \theta_2,C,\last(C)),\vect{\sigma})    =
\rho((\theta_1,C,\last(C)),\vect{\sigma})\wedge \rho((\theta_2,C,\last(C)),\vect{\sigma})$\vspace{0.2cm}

\item
  $
  \rho((\Next \theta,C,\last(C)),\vect{\sigma})    = ((\theta,C,\first(C)),\dir(\first(C)))
$ \vspace{0.2cm}

\item
  $
  \rho(( \theta_1\Until \theta_2,C,\last(C)),\vect{\sigma})    =
  \rho((\theta_2,C,\last(C)),\vect{\sigma}) \,\,\vee $

  \hspace{3cm} $[\rho((\theta_1,C,\last(C)),\vect{\sigma}) \wedge ((\theta_1\Until \theta_2,C,\first(C)),\dir(\first(C)))]
$ \vspace{0.2cm}

\item
  $
  \rho(( \theta_1\Release \theta_2,C,\last(C)),\vect{\sigma})    =
  \rho((\theta_2,C,\last(C)),\vect{\sigma}) \,\,\wedge $

  \hspace{3cm} $[\rho((\theta_1,C,\last(C)),\vect{\sigma}) \vee ((\theta_1\Release \theta_2,C,\first(C)),\dir(\first(C)))]
$ \vspace{0.2cm}

\item
  $\rho((\tpl{C'}\theta,C,\last(C)),\vect{\sigma})    = \rho( (\theta,C',\last(C')),\vect{\sigma})$
\end{itemize}

We first show that if the given quantifier-free formula $\psi$ in $\textit{NNF}$ is in the bounded fragment of $\CHLTL$, then $\Au_\psi$ is a
$(|\psi|+1)$-synchronous $\NAAWA$. By construction, whenever the automaton is in a state associated with a sub-formula $\theta$ of $\psi$, then
$\Au_\psi$ can move only to states associated with $\theta$ or with strict sub-formulas of $\theta$. In particular, each path in a
run of $\Au_\psi$ can be factorized into a finite number $\nu_1,\ldots,\nu_k$ of contiguous segments (with $\nu_k$ possibly infinite) such that for each $i\in [1,k]$, segment $\nu_i$ is associated with a sub-formula $\theta_i$ of $\psi$ and a context $C_i$ occurring in $\psi$, and the following holds, where the \emph{offset} of a position vector $\wp=(j_1,\ldots,j_n)$ in $\N^{n}$ is the maximum over the differences between pairs of components, i.e. $\max(\{j_\ell-j_{\ell'}\mid \ell,\ell'\in [1,n]\})$:
\begin{itemize}
 \item there is some occurrence of $\theta_i$ in $\psi$ which is in the scope of the context modality $\tpl{C_i}$;
  \item if $i<k$, then $\theta_{i+1}$ is a strict sub-formula of $\theta_i$;
  \item if either $C_i$ is global or the root modality of $\theta_i$ is not in $\{\Until,\Release\}$, then  the offset at each node along the segment $\nu_i$ and at the first node of $\nu_{i+1}$ if $i<k$ is at most the offset at the beginning of $\nu_i$ plus one.
\end{itemize}
\noindent Hence, if $\psi$ is in the bounded fragment of $\CHLTL$,  the offset at each node of a run is at most $|\psi|+1$ and the result directly follows.\vspace{0.1cm}

It remains to prove the correctness of the construction. For this, we need additional notation. For a position vector
$\wp=(j_1,\ldots,j_n)$ and a context $C\subseteq \{x_1,\ldots,x_n\}$, we denote by $\wp +_C  1$ the position vector $(j'_1,\ldots,j'_n)$ defined as follows
for all $i\in [1,n]$: $j'_i=j_i+1$ if $x_i\in C$ and $j'_i=j_i$  otherwise.

For a pointed trace assignment $\Pi= \{x_1 \mapsto (\pi_1,j_1),\ldots,x_n \mapsto (\pi_n,j_n)\}$ over
$\{x_1,\ldots,x_n\}$, we denote by $\vect{\pi}(\Pi)$ the $n$-tuple of traces $(\pi_1,\ldots,\pi_n)$ and by
$\wp(\Pi)$ the position vector $(j_1,\ldots,j_n)$.

Given an $n$-tuple $\vect{\pi}=(\pi_1,\ldots,\pi_n)$ of traces, a position vector $\wp=(j_1,\ldots,j_n)$, and a state $q$ of $\Au_\psi$,
a \emph{$(q,\wp)$-run  of $\Au_\psi$ over the input $\vect{\pi}$} is defined as a run of $\Au_\psi$ over
$\vect{\pi}$ but we require that the root of the run is labeled by the pair $(q,\wp)$ (i.e., initially, the automaton is in state $q$ and for each $i\in [1,n]$ reads the $(j_i+1)^{th}$ symbol of the trace $\pi_i$). By construction and a straightforward induction on the structure of the sub-formulas of $\psi$, we obtain the following result, hence,
Proposition~\ref{prop:FromContextHyperToAAWA} directly follows.\vspace{0.1cm}

\noindent \textbf{Claim.} Let $\Pi$ be a pointed trace assignment over $\{x_1,\ldots,x_n\}$ and $(\xi,C,\last(C))$ be a state of $\Au_\psi$. Then, there is an
accepting $((\xi,C,\last(C)),\wp(\Pi))$-run of $\Au_\psi$ over $\vect{\pi}(\Pi)$ \emph{if and only if} $(\Pi,C)\models \xi$.

\end{proof}

\bigskip

\subsection{\textbf{Proof of the claim in the proof of Theorem~\ref{theo:ComplexityBoundedCHLTL}}}\label{APP:ComplexityBoundedCHLTL}

In this section, we establish the following result (corresponding to the claim in the proof of Theorem~\ref{theo:ComplexityBoundedCHLTL}).


\begin{lemma}\label{lemma:lowerboundBoundedCHLTL} Let $\AP=\{p\}$ and $n>1$. One construct in time polynomial in $n$ a bounded $\CHLTL$ formula $\psi(x,y)$
with two free variables $x$ and $y$ and not containing universal quantifiers (hence, the quantifier alternation depth is $0$)  such that for all traces $\pi_x$ and $\pi_y$, $\{x \mapsto (\pi_x,0),y \mapsto (\pi_y,0)\}\models \psi(x,y)$ iff
\begin{compactitem}
  \item $p$ occurs exactly once on $\pi_x$ (resp., $\pi_y$);
  \item for each $i\geq$, $p\in \pi_x(i)$ iff $p\in \pi_y(i+n*2^{n}* 2^{2^{n}})$.
\end{compactitem}
\end{lemma}
\begin{proof}
We encode a yardstick of length $n*2^{n}* 2^{2^{n}}$ by using a $2^{n}$-bit counter for expressing integers in the range
$[0,2^{2^{n}}-1]$ and an $n$-bit counter for keeping track of the position (index) $i\in [0,2^{n}-1]$ of the $(i+1)^{th}$-bit of each valuation $v$ of the
$2^{n}$-bit counter. In particular, such a valuation $v\in [0,2^{2^{n}}-1]$ is encoded by a sequence, called \emph{$2$-block}, of $2^{n}$ sub-blocks of length $n$ where
for each $i\in [0,2^{n}-1]$, the $(i+1)^{th}$ sub-block encodes both the value and the index of the $(i+1)^{th}$-bit in the binary representation of $v$. To this purpose, we exploit  additional trace variables (existentially quantified):
\begin{compactitem}
  \item the trace variable $x_{bl_1}$ for encoding the bit indexes $i\in [0,2^{n}-1]$ of the $2^{n}$-bit counter by sequences of $n$ bits (\emph{$1$-blocks}) and the trace variable $x_{\#_1}$ for separating $1$-blocks: $\Rel{p}{x_{\#_1}}$ holds exactly at the beginning of each $1$-block;
  \item the trace variable $x_{bl_2}$ for encoding the bit values of the $2^{n}$-bit counter: in particular, $\Rel{p}{x_{bl_2}}$ can hold only at the beginning of
  a $1$-block;
  \item the trace variable $x_{\#_2}$ for separating $2$-blocks: $\Rel{p}{x_{\#_2}}$ holds exactly at the beginning of each $2$-block.
\end{compactitem}
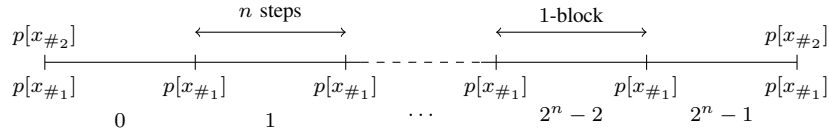
\begin{figure}[hbt!]
\centering
\vspace{-0.3cm}
\begin{tikzpicture}[scale=1]

\coordinate  (FirstTag) at (0.0,0.0);
\coordinate  (FirstTagAbove) at (0.0,0.1);
\coordinate  (FirstTagBelow) at (0.0,-0.1);
\coordinate  (SecondTag) at (2.0,0.0);
\coordinate  (SecondTagAbove) at (2.0,0.1);
\coordinate  (SecondTagBelow) at (2.0,-0.1);
\coordinate  (ThirdTag) at (4.0,0.0);
\coordinate  (ThirdTagAbove) at (4.0,0.1);
\coordinate  (ThirdTagBelow) at (4.0,-0.1);
\coordinate  (LeftMiddle) at (4.2,0.0);
\coordinate  (RightMiddle) at (5.8,0.0);
\coordinate  (FourthTag) at (6.0,0.0);
\coordinate  (FourthTagAbove) at (6.0,0.1);
\coordinate  (FourthTagBelow) at (6.0,-0.1);
\coordinate  (FifthTag) at (8.0,0.0);
\coordinate  (FifthTagAbove) at (8.0,0.1);
\coordinate  (FifthTagBelow) at (8.0,-0.1);
\coordinate  (LastTag) at (10.0,0.0);
\coordinate  (LastTagAbove) at (10.0,0.1);
\coordinate  (LastTagBelow) at (10.0,-0.1);

\path[-, thin,black] (FirstTagAbove) edge node[above] {\footnotesize $\Rel{p}{x_{\#_2}}$} (FirstTag);
\path[-, thin,black] (FirstTag) edge node[below] {\footnotesize $\Rel{p}{x_{\#_1}}$} (FirstTagBelow);

\coordinate [label=below:{\footnotesize  $0$}] (FirstValue) at (1.0,-0.55);

\path[-, thin,black] (SecondTagAbove) edge  (SecondTag);
\path[-, thin,black] (SecondTag) edge node[below] {\footnotesize $\Rel{p}{x_{\#_1}}$} (SecondTagBelow);

\path[<->, thin,black] (2.0,0.4) edge node[above] {\footnotesize $n$ steps} (4.0,0.4);
\path[<->, thin,black] (6.0,0.4) edge node[above] {\footnotesize $1$-block} (8.0,0.4);

\coordinate [label=below:{\footnotesize  $1$}] (SecondValue) at (3.0,-0.55);
\coordinate [label=below:{\footnotesize  $\ldots$}] (MiddleValue) at (5.0,-0.55);

\path[-, thin,black] (ThirdTagAbove) edge  (ThirdTag);
\path[-, thin,black] (ThirdTag) edge node[below] {\footnotesize $\Rel{p}{x_{\#_1}}$} (ThirdTagBelow);

\path[-, thin,black] (FourthTagAbove) edge  (FourthTag);
\path[-, thin,black] (FourthTag) edge node[below] {\footnotesize $p[x_{\#_1}]$} (FourthTagBelow);

\path[-, thin,black] (FifthTagAbove) edge  (FifthTag);
\path[-, thin,black] (FifthTag) edge node[below] {\footnotesize $\Rel{p}{x_{\#_1}}$} (FifthTagBelow);

\coordinate [label=below:{\footnotesize  $2^n-2$}] (LastButOneValue) at (7.0,-0.45);
\coordinate [label=below:{\footnotesize  $2^n-1$}] (LastValue) at (9.0,-0.45);

\path[-, thin,black] (LastTagAbove) edge node[above] {\footnotesize $\Rel{p}{x_{\#_2}}$} (LastTag);
\path[-, thin,black] (LastTag) edge node[below] {\footnotesize $\Rel{p}{x_{\#_1}}$} (LastTagBelow);

\path[-, thin,black] (FirstTag) edge (LeftMiddle);
\path[-, dashed,black] (LeftMiddle) edge (RightMiddle);
\path[-, thin,black] (RightMiddle) edge (LastTag);

\end{tikzpicture}
\vspace{-0.3cm}
\end{figure}
The formula $\psi(x,y)$ ensures that the segment between the unique
position where $\Rel{p}{x}$ holds and the unique position where
$\Rel{p}{y}$ holds (the latter excluded) is a sequence of $2^{2^{n}}$
$2$-blocks such that the first block encodes the $2^{n}$-counter
valuation $0$, and the $2^{n}$-counter is incremented on moving from a
$2$-block to the next one.
We assume that the first bit of the $n$-bit counter (resp.,
$2^{n}$-bit counter) is the least significant one.
Moreover, in order to ensure that the $2^{n}$-counter is correctly
updated, we exploit $n$ additionally existentially quantified trace
variables $x_1,\ldots,x_n$ where for each $j\in [1,n]$, $\Rel{p}{x_j}$
holds exactly at the beginning of each $1$-block.
The meaning of such variables will be explained later.
\begin{figure}[hbt!]
\centering
\vspace{-0.3cm}
\begin{tikzpicture}[scale=1]

\coordinate  (FirstTag) at (0.0,0.0);
\coordinate  (FirstTagAbove) at (0.0,0.1);
\coordinate  (FirstTagBelow) at (0.0,-0.1);
\coordinate  (SecondTag) at (2.0,0.0);
\coordinate  (SecondTagAbove) at (2.0,0.1);
\coordinate  (SecondTagBelow) at (2.0,-0.1);
\coordinate  (ThirdTag) at (4.0,0.0);
\coordinate  (ThirdTagAbove) at (4.0,0.1);
\coordinate  (ThirdTagBelow) at (4.0,-0.1);
\coordinate  (LeftMiddle) at (4.2,0.0);
\coordinate  (RightMiddle) at (5.8,0.0);
\coordinate  (FourthTag) at (6.0,0.0);
\coordinate  (FourthTagAbove) at (6.0,0.1);
\coordinate  (FourthTagBelow) at (6.0,-0.1);
\coordinate  (FifthTag) at (8.0,0.0);
\coordinate  (FifthTagAbove) at (8.0,0.1);
\coordinate  (FifthTagBelow) at (8.0,-0.1);
\coordinate  (LastTag) at (10.0,0.0);
\coordinate  (LastTagAbove) at (10.0,0.1);
\coordinate  (LastTagBelow) at (10.0,-0.1);

\path[-, thin,black] (FirstTagAbove) edge node[above] {\footnotesize $\Rel{p}{x}$} (FirstTag);
\path[-, thin,black] (FirstTag) edge node[below] {\footnotesize $\Rel{p}{x_{\#_2}}$} (FirstTagBelow);

\coordinate [label=below:{\footnotesize  $0$}] (FirstValue) at (1.0,-0.55);

\path[-, thin,black] (SecondTagAbove) edge  (SecondTag);
\path[-, thin,black] (SecondTag) edge node[below] {\footnotesize $\Rel{p}{x_{\#_2}}$} (SecondTagBelow);

\path[<->, thin,black] (2.0,0.4) edge node[above] {\footnotesize $n* 2^{n}$ steps} (4.0,0.4);
\path[<->, thin,black] (6.0,0.4) edge node[above] {\footnotesize $2$-block} (8.0,0.4);

\coordinate [label=below:{\footnotesize  $1$}] (SecondValue) at (3.0,-0.55);
\coordinate [label=below:{\footnotesize  $\ldots$}] (MiddleValue) at (5.0,-0.55);

\path[-, thin,black] (ThirdTagAbove) edge  (ThirdTag);
\path[-, thin,black] (ThirdTag) edge node[below] {\footnotesize $\Rel{p}{x_{\#_2}}$} (ThirdTagBelow);

\path[-, thin,black] (FourthTagAbove) edge  (FourthTag);
\path[-, thin,black] (FourthTag) edge node[below] {\footnotesize $\Rel{p}{x_{\#_2}}$} (FourthTagBelow);

\path[-, thin,black] (FifthTagAbove) edge  (FifthTag);
\path[-, thin,black] (FifthTag) edge node[below] {\footnotesize $\Rel{p}{x_{\#_2}}$} (FifthTagBelow);

\coordinate [label=below:{\footnotesize  $2^{2^{n}}-2$}] (LastButOneValue) at (7.0,-0.45);
\coordinate [label=below:{\footnotesize  $2^{2^{n}}-1$}] (LastValue) at (9.0,-0.45);

\path[-, thin,black] (LastTagAbove) edge node[above] {\footnotesize $\Rel{p}{y}$} (LastTag);
\path[-, thin,black] (LastTag) edge node[below] {\footnotesize $\Rel{p}{x_{\#_2}}$} (LastTagBelow);

\path[-, thin,black] (FirstTag) edge (LeftMiddle);
\path[-, dashed,black] (LeftMiddle) edge (RightMiddle);
\path[-, thin,black] (RightMiddle) edge (LastTag);

\end{tikzpicture}
\vspace{-0.3cm}
\end{figure}
Formally, the bounded $\CHLTL$ formula $\psi(x,y)$ is defined as follows:
\[
\begin{array}{l}
\psi(x,y)\DefinedAs \exists x_{\#_1}.\,\exists x_{\#_2}.\,\exists x_{bl_1}.\,\exists x_{bl_2}.\,\exists x_{1}.\dots \exists x_{n}.\,\\
\hspace{2cm}\psi_{unique}\wedge \psi_{align}\wedge \psi_{con}\wedge \psi_{bl_2}\wedge\psi_{first}\wedge \psi_{last}\wedge \psi_{inc}
\end{array}
\]
where the various conjuncts are quantifier-free formulas defined in
the following (in particular, all the conjuncts but $\psi_{inc}$ are
$\HLTL$ quantifier-free formulas).  For each trace variable $z$ and
$b\in\{0,1\}$, we write $\Rel{p}{z} = b$ to mean the formula $\Rel{p}{z}$ if
$b=1$, and the formula $\neg \Rel{p}{z}$ otherwise.\vspace{0.1cm}

\noindent \textbf{Definition of $\psi_{unique}$.} $\psi_{unique}$
requires that $\Rel{p}{x}$ and $\Rel{p}{y}$ hold exactly once with
$\Rel{p}{x}$ strictly preceding $\Rel{p}{y}$.
\[
  \psi_{unique} \DefinedAs \Eventually(\Rel{p}{x}\wedge \Next\Eventually
  \Rel{p}{y})\wedge \displaystyle{\bigwedge_{z\in\{x,y\}}}
  \Always(\Rel{p}{z}\rightarrow \Next\Always(\neg \Rel{p}{z}))
\]

\noindent \textbf{Definition of $\psi_{align}$.} $\psi_{align}$
requires that
\begin{inparaenum}[(i)]
\item at the positions where $\Rel{p}{x}$ and $\Rel{p}{y}$
  hold, $\Rel{p}{x_{\#_1}}$ and $\Rel{p}{x_{\#_2}}$ hold as well,
\item $\Rel{p}{x_{\#_1}}$ and $\Rel{p}{x_{\#_2}}$ do not hold at
  positions preceding $\Rel{p}{x}$ and following $\Rel{p}{y}$,
\item $\Rel{p}{x_{bl_2}}$ can hold
  only where $\Rel{p}{x_{\#_1}}$ holds, and
\item for each $i\in [1,n]$, $\Rel{p}{x_i}$ holds exactly where
  $\Rel{p}{x_{\#_1}}$ holds.
\end{inparaenum}
\[ \begin{array}{l}
\psi_{align} \DefinedAs \displaystyle{\bigwedge_{j\in \{1,2\}}} \neg \Rel{p}{x_{\#_j}}\Until \bigl(\Rel{p}{x_{\#_j}}\wedge \Rel{p}{x}\wedge \Eventually(\Rel{p}{x_{\#_j}}\wedge \Rel{p}{y}\wedge \Next\Always \neg \Rel{p}{x_{\#_j}} ) \bigr)\wedge   \vspace{0.2cm}\\
 \hspace{1.5cm} \Always(\Rel{p}{x_{bl_2}}\rightarrow \Rel{p}{x_{\#_1}})\wedge \displaystyle{\bigwedge_{i=1}^{i=n}}\Always(\Rel{p}{x_i} \leftrightarrow \Rel{p}{x_{\#_1}})
\end{array} \]
\noindent \textbf{Definition of $\psi_{con}$.} $\psi_{con}$ ensures that the segment between $\Rel{p}{x}$ and $\Rel{p}{y}$ (the $\Rel{p}{y}$-position excluded) is the concatenation of segments of length $n$ (1-blocks) and $p[x_{\#_1}]$ holds exactly at the beginning of each of such segments.
\[
\psi_{con} \DefinedAs \Always(p[x_{\#_1}] \rightarrow [\displaystyle{\bigwedge_{i=1}^{i=n-1}}\Next^{i}\neg \Rel{p}{x_{\#_1}} \wedge (\Rel{p}{y}\vee \Next^{n}\Rel{p}{x_{\#_1}})]
\]
\noindent \textbf{Definition of $\psi_{bl_2}$.} $\psi_{bl_2}$ ensures
that each segment $\nu$ between two consecutive
$\Rel{p}{x_{\#_2}}$-positions (the last position of $\nu$ excluded)
encodes a $2$-block, i.e.
\begin{inparaenum}[(i)]
\item the first $1$-block of $\nu$ encodes $0$,
\item the last $1$-block of $\nu$ encodes $2^{n}-1$, and (iii) in
  moving along $\nu$ from a $1$-block to the next one, the $n$-bit
  counter is incremented.
\end{inparaenum}
\[ \begin{array}{l} \psi_{bl_2} \DefinedAs
    \Always(\Rel{p}{x_{\#_2}}\rightarrow \Rel{p}{x_{\#_1}}) \wedge
    \Always[(\Rel{p}{x_{\#_2}}\wedge \neg \Rel{p}{y}) \rightarrow
     (\displaystyle{\bigwedge_{i=0}^{i=n-1}}\Next^{i}\neg \Rel{p}{x_{bl_1}}\wedge \theta_{last})] \wedge \theta_{inc}  \\
     \theta_{last}\DefinedAs \Next [\neg \Rel{p}{x_{\#_2}}\,\Until\, (\neg \Rel{p}{x_{\#_2}}\wedge \Rel{p}{x_{\#_1}} \wedge \displaystyle{\bigwedge_{i=0}^{i=n-1}} \Next^{i} \Rel{p}{x_{bl_1}}\wedge  \Next^{n}\Rel{p}{x_{\#_2}})]\vspace{0.2cm}\\
     \theta_{inc}\DefinedAs \Always \bigl[\bigl(\Rel{p}{x_{\#_1}}\wedge \Next^{n}(\Rel{p}{x_{\#_1}}\wedge \neg \Rel{p}{x_{\#_2}})\bigr) \longrightarrow \vspace{0.2cm}\\
     \hspace{1cm}\displaystyle{\bigvee_{i=0}^{i=n-1}}
     \bigl( \Next^{i} \theta(0,1) \wedge \displaystyle{\bigwedge_{j=0}^{j=i-1}} \Next^{j} \theta(1,0) \wedge \displaystyle{\bigwedge_{j=i+1}^{j=n-1}\bigvee_{b\in \{0,1\}}} \Next^{j} \theta(b,b)    \bigr)\bigr]\vspace{0.2cm}\\
     \theta(b,b')\DefinedAs \Rel{p}{x_{bl_1}}= b \wedge \Next^{n}(\Rel{p}{x_{bl_1}}= b')
\end{array} \]
\noindent \textbf{Definition of $\psi_{first}$.} $\psi_{first}$ requires that the first $2$-block encodes $0$.
\[
  \psi_{first} \DefinedAs \Always\bigl(\Rel{p}{x}\rightarrow (\neg
  \Rel{p}{x_{bl_2}} \wedge \Next\{(\neg \Rel{p}{x_{\#_2}}\wedge
  (\Rel{p}{x_{\#_1}} \rightarrow \neg \Rel{p}{x_{bl_2}}))\,\Until\,
  \Rel{p}{x_{\#_2}} \} )\bigr) \]
\noindent \textbf{Definition of $\psi_{last}$.} $\psi_{last}$ requires that the last $2$-block encodes $2^{2^{n}}-1$.
\[
  \psi_{last} \DefinedAs \Always\bigl(\Next(\neg \Rel{p}{x_{\#_2}}\,\Until\,
  \Rel{p}{y})\longrightarrow (\Rel{p}{x_{\#_1}} \rightarrow
  \Rel{p}{x_{bl_2}})\bigr)
\]
\noindent \textbf{Definition of $\psi_{inc}$.}
Formula $\psi_{inc}$ ensures that in moving from a non-last $2$-block
$bl$ to the next one $bl'$, the $2^{n}$-counter is incremented.
This is equivalent to require that there is a $1$-block $sbl_0$ of
$bl$ whose bit value (i.e., the Boolean value of $\Rel{p}{x_{bl_2}}$
at the beginning of $sbl_0$) is $0$ such that for each $1$-block $sbl$
of $bl$, denoted by $sbl'$ the $1$-block of $bl'$ having the same
$n$-bit counter valuation as $sbl$, the following holds: (i) if $sbl$
precedes $sbl_0$, then the bit value of $sbl$ (resp., $sbl'$) is $1$
(resp., $0$), (ii) if $sbl$ corresponds to $sbl_0$, then the bit value
of $sbl'$ is $1$, and (iii) if $sbl$ follows $sbl_0$, then there is
$b\in\{0,1\}$ such that the bit value of both $sbl$ and $sbl'$ is $b$.
The definition of $\psi_{inc}$ is the crucial part of the
construction, where we exploit the context modalities (in a bounded
way) and the trace variables $x_1,\ldots,x_n$.
The definition of $\psi_{inc}$ is based on the construction of
auxiliary formulas $\psi_=(b,b')$ where $b,b'\in\{0,1\}$.
Formula $\psi_=(b,b')$ holds at a position $h$ (along all the current
traces) iff whenever $h$ corresponds to the beginning of a $1$-block
$sbl$ of a $2$-block $bl$, then (i) the bit value of $sbl$ is $b$,
(ii) the $2$-block $bl$ is followed by a $2$-block $bl'$, and (iii)
the $1$-block of $bl'$ having the same $n$-bit counter valuation as
$sbl$ has bit value $b'$.
For expressing this requirement, from the current position $h$ (along
all the current traces) associated to the first position of a
$1$-block $sbl$ (a $\Rel{p}{x_{\#_1}}$-position) of a non-last
$2$-block $bl$, we exploit the non-global context modality
$\tpl{\{x_i\}}$, for each $i\in [1,n]$, for ensuring that the local
position along the trace for $x_i$ moves one position to the right if
and only if the $i^{th}$-bit of $sbl$ along the trace for $x_{bl_1}$
is $0$ (recall that $\Rel{p}{x_i}$ holds exactly at the first
positions of $1$-blocks and $n>1$).
Then, by exploiting the global context modality and temporal
modalities, we require that the local positions of all the traces move
forward synchronously step by step until the local position of the
trace for variable $x_{bl_1}$ corresponds to the first position of a
$1$-block $sbl'$ of the $2$-block $bl'$ following $bl$.
At this point, in order to check that $sbl$ and $sbl'$ have the same
counter valuation, it suffices to check that the current Boolean value
of $\Rel{p}{x_i}$ corresponds to the $i^{th}$ bit of $sbl'$ for each
$i\in [1,n]$.

xIn order to define $\psi_=(b,b')$, we first define the auxiliary
quantifier-free formulas $\theta_1(b),\ldots,\theta_n(b)$ by induction
on $n-i$ where $i\in [0,n-1]$.
\[ \begin{array}{l}
     \theta_n(b) \DefinedAs \tpl{\{x_n\}}[ \Next\bigl(\xi(b) \wedge \tpl{\{x_{bl_1}\}}\Next^{n-1}\neg \Rel{p}{x_{bl_1}}\bigr)\vee \bigl(\xi(b) \wedge \tpl{\{x_{bl_1}\}}\Next^{n-1} \Rel{p}{x_{bl_1}}\bigr)]\vspace{0.2cm}\\
     \xi(b)\DefinedAs \tpl{\Var} \Next\Bigl(\neg \Rel{p}{x_{\#_2}}\Until \bigl(\Rel{p}{x_{\#_2}} \wedge (\chi(b)\vee \Next(\neg \Rel{p}{x_{\#_2}}\Until (\neg \Rel{p}{x_{\#_2}}\wedge \chi(b)) ) ) \bigr)\Bigr)\vspace{0.1cm}\\
     \chi(b) \DefinedAs \Rel{p}{x_{\#_1}} \wedge \Rel{p}{x_{bl_2}} = b \wedge
     \displaystyle{\bigwedge_{i=1}^{i=n}}(\Rel{p}{x_i}\leftrightarrow
     \Next^{i-1} \Rel{p}{x_{bl_1}}x )
   \end{array} \]
 Moreover, for each $j\in [1,n-1]$, formula $\theta_j(b)$ is defined
 as follows:
 \[
 \theta_j(b) \DefinedAs \tpl{\{x_j\}}[ \Next\bigl(\theta_{j+1}(b) \wedge \tpl{\{x_{bl_1}\}}\Next^{j-1}\neg \Rel{p}{x_{bl_1}}\bigr)\vee \bigl(\theta_{j+1}(b) \wedge \tpl{\{x_{bl_1}\}}\Next^{j-1} \Rel{p}{x_{bl_1}}\bigr)]
 \]
Then, for all $b,b'\in \{0,1\}$, the quantifier-free formula $\psi_=(b,b')$ is given by
\[
\psi_=(b,b')\DefinedAs \Rel{p}{x_{\#_1}} \rightarrow (\theta_1(b')\wedge \Rel{p}{x_{bl_2}}=b)
\]
Finally, formula $\psi_{inc}$ is defined as follows:
\[ \begin{array}{l}
     \psi_{inc}  ::=  \Always\Bigl[\bigl(\Rel{p}{x_{\#_2}}\wedge \Next\Eventually (\Rel{p}{x_{\#_2}} \wedge \neg \Rel{p}{y})\bigr) \longrightarrow \Bigl((\psi_=(0,1)\wedge \eta) \vee \\
     \hspace{1cm} \bigl(\psi_=(1,0)\wedge \Next((\neg \Rel{p}{x_{\#_2}} \wedge \psi_=(1,0))  \,\Until\, (\neg \Rel{p}{x_{\#_2}}\wedge \Rel{p}{x_{\#_1}}\wedge \psi_=(0,1) \wedge \eta  ))   \bigr)\Bigr)\Bigr] \vspace{0.2cm}\\
     \eta ::= \Next((\neg \Rel{p}{x_{\#_2}} \wedge
     \displaystyle{\bigvee_{b\in\{0,1\}}} \psi_=(b,b) ) \,\Until\,
     \Rel{p}{x_{\#_2}})
   \end{array} \]
 Note that the number of distinct sub-formulas of $\psi(x,y)$ is
 polynomial in $n$.
\end{proof}
\end{changemargin}


\end{document}